\documentclass[a4paper,reqno]{amsart}
\usepackage[text={5.2in,8in},centering]{geometry}
\usepackage{amsrefs}
\usepackage{amssymb}
\usepackage{mathrsfs}

\usepackage{xspace}

\DeclareMathAlphabet{\mathpzc}{OT1}{pzc}{m}{it}

\usepackage{enumerate}
\usepackage{graphicx}
\usepackage{color}
\definecolor{trustcolor}{rgb}{0.71,0.14,0.07}

\numberwithin{equation}{section}
\theoremstyle{plain}
\newtheorem{theorem}{Theorem}[section]

\newtheorem{lemma}{Lemma}[section]

\theoremstyle{remark}
\newtheorem{remark}{Remark}[section]

\newtheorem*{quest*}{Question}
\newtheorem*{remark*}{Remark}
\theoremstyle{remark}

\theoremstyle{definition}
\newtheorem{definition}{Definition}[section]
\newtheorem*{definition*}{Definition}
\newtheorem*{notation*}{Notation}
\newtheorem*{notations*}{Notations}
\providecommand{\B}{\mathbf}
\providecommand{\bS}[1]{\boldsymbol{#1}}
\providecommand{\C}{\mathcal}
\providecommand{\CS}{\mathscr}

\providecommand{\D}{\mathbb}

\providecommand{\R}{\mathrm}

\newcommand{\ee}{\mathrm{e}}
\newcommand{\eul}{\mathrm{e}}
\newcommand{\ii}{\mathrm{i}}

\def\boxx{\mathrm{B}}
\def\ball{\mathrm{B}}

\def\bball{\mathbf{B}}

\DeclareMathOperator{\dist}{dist}
\DeclareMathOperator{\card}{card}

\DeclareMathOperator*{\essup}{ess\,sup}
\DeclareMathOperator*{\supp}{supp}
\DeclareMathOperator{\expect}{\mathbb{E}}

\DeclareMathOperator{\one}{\mathbf{1}}

\DeclareMathOperator{\diam}{{\rm diam}}

\def\loc{$m\mathcal{-L}${\textit{oc}}}
\def\nloc{$m\mathcal{-N}\mathcal{L}${\textit{oc}}}

\def\mathnloc{m\mathcal{-N}\mathcal{L}{\textit{oc}}}

\def\mT{{\rm $m$-T}\xspace}
\def\NT{{\rm $m$-NT}\xspace}

\def\ER{{\rm $E$-R}\xspace}
\def\NR{{\rm $E$-NR}\xspace}
\def\CNR{{\rm $E$-CNR}\xspace}
\def\PR{{\rm $E$-PR}\xspace}

\def\EmS{{\rm $(E,m)$-S}\xspace}
\def\EmNS{{\rm $(E,m)$-NS}\xspace}

\def\lam{{\lambda}}
\def\eps{\epsilon}

\def\Bn{\B{n}}
\def\Bx{\B{x}}
\def\By{\B{y}}
\def\Bu{\B{u}}
\def\Bv{\B{v}}
\def\Bw{\B{w}}
\def\Bz{\B{z}}

\def\tf{\tilde{f}}

\def\vBx{\vec{\B{x}}}

\def\BPsi{{\bS{\Psi}}}
\def\BPhi{{\bS{\Phi}}}

\def\Const{{\rm{Const}}}

\def\Bpsi{{\bS{\psi}}}
\def\Bphi{{\bS{\varphi}}}

\def\Lam{{\Lambda}}
\def\lam{{\lambda}}

\def\Ba{\B{a}}
\def\Bb{\B{b}}

\def\BA{\B{A}}
\def\BB{\B{B}}

\def\BG{\B{G}}
\def\BH{\B{H}}
\def\BK{\B{K}}

\def\BP{\B{P}}
\def\BS{\B{S}}
\def\BU{\B{U}}
\def\BV{\B{V}}
\def\BW{\B{W}}
\def\BX{\B{X}}
\def\BDelta{\B{\Delta}}

\def\BLam{\B{\Lambda}}

\def\bcH{\boldsymbol{\C{H}}}
\def\bcHN{\boldsymbol{\C{H}^N}}

\def\uw{\B{w}}
\def\DC{\D{C}}
\def\DD{\D{D}}
\def\DP{\D{P}}
\def\DR{\D{R}}
\def\DZ{\D{Z}}
\def\DN{\D{N}}

\def\cD{\C{D}}
\def\cE{\C{E}}
\def\cF{\C{F}}
\def\cG{\C{G}}
\def\cH{\C{H}}
\def\cJ{\C{J}}
\def\cL{\C{L}}
\def\cM{\C{M}}
\def\cN{\C{N}}

\def\cR{\C{R}}
\def\cS{\C{S}}

\def\cX{\C{X}}
\def\csB{\CS{B}}

\def\cZp{\boldsymbol{\C{Z}^N_{>}}}
\def\cZpzero{\boldsymbol{\C{Z}^N_{=}}}
\def\cZpeq{\boldsymbol{\C{Z}^N_{\ge}}}
\def\bcEp{\boldsymbol{\C{E}^N_{>}}}

\def\cZN{\boldsymbol{\C{Z}^N}}

\def\mytimes{\operatornamewithlimits{\hbox{\huge$\times$}}}

\def\cZ{\C{Z}}

\def\rd{{\R{d}}}
\def\rdS{{{\rho_S}}}
\def\rdN{{\R{d_{\cZp}}}}
\def\rc{{\R{c}}}

\def\hx{\hat{x}}
\def\hBx{\hat{\mathbf{x}}}

\def\tm{ {\widetilde{m}}}

\def\jknot{{j_\circ}}

\def\muxy{{\mu^{\Bx,\By}_{\om}}}

\def\bmubxy{{\mu^{\Bx,\By}_{\bball,\om}}}
\def\bmukxy{{\mu^{\Bx,\By}_{\bball_{L_k},\om}}}
\def\bmuxy{{\mu^{\Bx,\By}_{\om}}}

\def\be{\begin{equation}}
\def\ee{\end{equation}}
\def\ba{\begin{array}{l}}
\def\ea{\end{array}}
\def\bal{\begin{aligned}}
\def\eal{\end{aligned}}

\def\fF{\mathfrak{F}}
\def\fS{\mathfrak{S}}

\def\om{{\omega}}
\def\Om{{\Omega}}

\def\pr#1{\D{P}\left\{\,#1\,\right\}}
\def\esm#1{\D{E}\left[\, #1\, \right]}
\def\pt{\partial}
\def\half{\frac{1}{2}}

\def\quart{\frac{1}{4}}

\def\nb#1{{ \langle #1 \rangle}}

\def\truc#1#2#3{\smash{\mathop{\,\, #1 \,\, }\limits^{#2}_{#3}}}

\def\tto#1{\smash{\mathop{\,\,\,\, \longrightarrow \,\,\,\, }\limits_{#1}}}

\def\myset#1{{\left\{\,#1\,\right\}}}

\def\ketbra#1#2{{ | {#1} \rangle \langle {#2}| }}
\def\bra#1{{ \langle {#1}| }}
\def\ket#1{{ | {#1} \rangle }}

\def\mymax#1{{ \truc{\max} {} {#1}}}

\def\diy{\displaystyle}

\def\Uzero{{\rm\textsf{U0}}\xspace}
\def\Uone{{\rm\textsf{U1}}\xspace}

\def\Wone{{\rm\textsf{W1}}\xspace}
\def\Wtwo{{\rm\textsf{W2}}\xspace}
\def\Wthree{{\rm\textsf{W3}}\xspace}

\begin{document}

\title[Direct Scaling Analysis. II. Multi-particle systems]
{Direct Scaling Analysis \\of localization in disordered systems.\\ II.  Multi-particle lattice systems}

\author[V. Chulaevsky]{Victor Chulaevsky}


\address{D\'{e}partement de Math\'{e}matiques\\
Universit\'{e} de Reims, Moulin de la Housse, B.P. 1039\\
51687 Reims Cedex 2, France\\
E-mail: victor.tchoulaevski@univ-reims.fr}

\date{}
\begin{abstract}
We adapt a simplified version of the Multi-Scale Analysis presented in \cite{C11} to multi-particle tight-binding Anderson models. Combined with a recent eigenvalue concentration bound for multi-particle systems \cite{C10}, the new method leads to a simpler proof of the multi-particle dynamical localization with more optimal decay bounds on eigenfunctions than   in \cite{CS09b,AW09a,AW09b}, for a large class of strongly mixing random potentials. All earlier results required the random potential to be IID. We also extend the result on multi-particle  localization to models with a rapidly decaying interaction.

\textbf{Note: this is an improved version of the manuscript originally uploaded on 11.06.2011.}

\end{abstract}

\maketitle
\section{Introduction. The model, assumptions and the main results} \label{sec:intro}

Analysis of localization phenomena in multi-particle quantum systems with nontrivial interaction in a random environment is a relatively new direction in the Anderson localization theory, where during half a century, since the seminal paper by P. Anderson \cite{A58}, most efforts were concentrated on the study of disordered systems in single-particle approximation, i.e., without inter-particle interaction. While in numerous physical models such an approximation is fairly reasonable, it was pointed out already in the first works  by Anderson that multi-particle models presented a real challenge.

The number of results on multi-particle localization obtained both by physicists and mathematicians remains yet rather limited. We do not review here results obtained by physicists,
based on methods of theoretical physics and considered as firmly established by the physical community.
The rigorous mathematical results on multi-particle localization obtained so far (cf. \cite{CS09a,CS09b,CBS11}, \cite{AW09a,AW09b}) apply to $N$-particle systems with arbitrary, but fixed $N>1$, and the range of parameters (such as the amplitude of the disorder and/or proximity to the edge(s) of the spectrum) is rapidly degrading as $N\to\infty$.

In this paper, we study spectral properties of random lattice Schr\"{o}dinger operators (LSO)  in the general framework of the Multi-Scale Analysis  (MSA) developed in \cite{FMSS85,Spe88,DK89,DK91}.
Specifically, we study random Hamiltonians of the form
\be\label{eq:H}
\BH(\om) = \BH_0 + g\BV(\om) + \BU, \;g\in\DR,
\ee
where $\BH_0$ is a finite-difference operator representing the kinetic energy, e.g., the nearest-neighbor lattice (negative) Laplacian $(-\BDelta)$, $\BU$ is the operator of multiplication by the interaction potential $\Bx\mapsto\BU(\Bx)$, and $\BV(\om)$ is the operator of multiplication by
the function $\Bx\mapsto \BV(\Bx;\om) = V(x_1;\om) + \cdots  V(x_N;\om)$,
where $V:\DZ^d\times\Om\to\DR$ is a random field relative to some probability space
$(\Om,\cF,\DP)$.
As was shown in \cite{CS09b}, \cite{AW09a,AW09b}, Anderson localization can be established in the
entire space $\cH_N$, but in this paper we consider only the restriction of $\BH(\om)$ to the
fermionic subspace. The bosonic subspace can be treated essentially in the same way. Using either
kind of quantum statistics makes the scaling analysis more transparent and proofs substantially
simpler.

Aizenman and Warzel \cite{AW09a,AW09b} established the exponential dynamical localization in multi-particle systems in the framework of a parametric analysis allowing to prove
local stability (in the parameter space) of the Anderson localization phenomenon (including
exponential strong dynamical localization) for Anderson-type Hamiltonians with an IID external random potential under perturbations by a short-range interaction.

Apart from simplifications of the Multi-Particle Multi-Scale Analysis (MPMSA) developed in
\cite{CS09a,CS09b}, the novelty of this paper is three-fold:
\begin{itemize}
  \item the external random potential is not assumed to be IID, but is strongly mixing;
  \item the interaction potential is not necessarily of finite range, but can be exponentially or even sub-exponentially fast decaying at infinity;
  \item the decay of eigenfunctions is proven with respect to a norm-distance in $\DZ^{Nd}$
  (more precisely, in a symmetrized norm-distance), while earlier results used explicitly (\cite{AW09a}) or implicitly (\cite{CS09b}) a more complicated and less natural notion of distance in the multi-particle configuration space\footnote{A detailed  analysis of the nature of this technical difficulty, performed by Aizenman and Warzel \cite{AW09a}, proved instrumental for a solution proposed later in \cite{C10}.}.
\end{itemize}

Note that our proof of multi-particle localization for strongly mixing random potentials and rapidly decaying infinite-range interaction can be extended to the entire Hilbert space $\cH_N$, at the price of greater complexity. On the other hand, the proofs can be made simpler and more streamlined, if only a \emph{power-law} decay of the eigenfunction correlators were to be established (as is customary in the framework of the MSA); the latter is not to be confused with the decay of \emph{eigenfunctions} which is exponential.

We consider here only strongly disordered systems; an adaptation of our approach to localization at "extreme" energies in weakly disordered multi-particle systems, as well as to weak perturbations of localized non-interacting multi-particle systems, has been recently obtained by Ekanga \cite{E11} in the framework of his PhD project.

The new scaling procedure used in this paper is an adaptation of the method proposed earlier in \cite{C11} in the single-particle context; it is close in spirit to a very simple method developed by Spencer \cite{Spe88} for a fixed-energy analysis of Green functions\footnote{I thank Tom Spencer, Abel Klein and Boris Shapiro for fruitful discussions of the work \cite{Spe88}.}.

The proofs of all statements not given in the main text can be found in Appendix.

The structure of this paper is as follows.

\begin{itemize}

  \item In Section \ref{sec:basic.notions} we describe main geometrical notions and constructions relative to fermionic systems; for brevity of presentation, we
      consider first the case of one-dimensional particles. Nevertheless,
      our scheme in Sections \ref{sec:deterministic.bounds}--\ref{sec:SimpleDL.finite.L}
      is not specific to one dimension; indeed, it has to be stressed that from the
      analytic point of view, an $N$-particle system in $\DZ^1$ (starting already
      with $N=2$) gives rise to a multi-dimensional (discrete) Schr\"{o}dinger operator.
  \item In Section \ref{sec:deterministic.bounds} we describe the analytic core of the new scaling procedure, which is essentially the same as in the single-particle case treated in \cite{C11}.

  \item The probabilistic inductive bounds for the model with a finite-range interaction are given in Section \ref{sec:MPMSA.finite.range}. The key statements here are Lemma \ref{lem:PI}, Lemma \ref{lem:DS.mixed} and Theorem \ref{thm:loc.ind}. A reader familiar with \cite{CS09b} may notice that an important component of the the proof of Lemma \ref{lem:PI}, given by Lemma \ref{lem:instead.PITRONS}, is much simpler than its counterpart in \cite{CS09b}, due to the use of eigenfunctions instead of Green functions.

  \item An adaptation to infinite-range interactions is given in Section \ref{sec:MPMSA.infinite.range}.

  \item The derivation of the strong dynamical localization from the results of the scaling analysis is given in Section \ref{sec:SimpleDL.finite.L}, where we follow the same path as in \cite{C11}. A reader familiar, e.g., with \cite{GD98,DS01} may want to skip Section \ref{sec:SimpleDL.finite.L}.

  \item An adaptation to more general graphs is discussed in Section \ref{sec:general.graph}.

  \item For the reader's convenience, we prove in Appendix the new eigenvalue concentration (EVC) bound  (Theorem \ref{thm:PCT}; cf. also \cite{C10}), allowing to establish the exponential decay of eigenfunctions with respect to the max-norm in the $N$-particle configuration space.
\end{itemize}

To conclude the introduction, note that we do not discuss here the ergodicity issues
for the Hamiltonians $\BH(\om)$, for several reasons.

First of all, it is well-known already in the single-particle localization theory,
where the Hamiltonian has the form $H(\om) = H_0 + gV(\om)$, say, on the lattice $\DZ^d$,
the ergodicity of the random field $V:\DZ^d\times\Om\to\DR$  per se is irrelevant for the usual
proofs of complete localization (pure point spectrum) under the assumption of strong disorder
($|g|\gg 1$). In the case of an IID potential, the location of the a.s. spectrum can be described
withthe help of the Weyl argument, again, without using explicitly the ergodicity (which
is, of course, granted for an IID random field).

The second reason is that already in the case of a two-particle Hamiltonian
$\BH(\om) = \BH_0 + \BW(\Bx;\om) = \BH_0 + \BV(\Bx;\om) + \BU(\Bx)$, the random field
$$
\big( (x_1, x_2), \om) \mapsto V(x_1;\om) + V(x_1;\om) + \BU(x_1,x_2)
$$
is no longer stationary. Specifically, let $d=1$ and $V(x;\om)\in[0,1]$, say,
uniformly distributed
in $[0,1]$, and $\BU(x_1,x_2) = U \one_{\{x_1=x_2\}}$, $U < -2$. Then
for all points of the form $\Bx=(x,x)$, $\BW(\Bx;\om) < 0$, while for all others
$\BW(\Bx;\om) \ge 0$. The translation invariance of the field $\BW(\Bx;\om)$ holds
only for a subrgoup of "diagonal" shifts $(x_1, x_2)\mapsto (x_1+a, x_2+a)$,
$a\in\DZ$.

The third, more important and less formal reason is that the language of the density
of states (DoS), very helpful and instructive in the context of single-particle models,
is much less so in the framework of multi-particle models with a nontrivial interaction.
Recall that, according to an earlier result by Klopp and Zenk \cite{KZ03}, proven for
multi-particle Schr\"{o}dinger operators in $\DR^d$ with decaying interaction, the DoS
is the same for operators with and without interaction. (An adaptation of their techniques
to the lattice is not difficult.) Consider again the two-particle operator $\BH(\om)$
from the previous paragraph; let $\BH_0\ge 0$. Then the \emph{density of states} is supported by the non-negative half-line $\DR_+$, but it is obvious that, with $U < -\|\BH_0\| - \|\BV\|$, the \emph{spectrum} has a negative component,
for the quadratic form $\langle f| \BH f\rangle$ is not sign-definite, having positive
and negative diagonal matrix elements in the delta-basis. Speaking informally, the DoS measure
indicates the location of the "bulk" spectrum, while the negative interaction gives rise
to some "internal surface" spectrum due to eigenfunctions (square-summable or generalized)
essentially supported by a neighborhood of $\supp \BU$ and decaying
away from it (which can be seen, e.g., with the help of the Combes--Thomas argument).

The sign of the interaction is not crucial for the above observation; this can be seen already from the fact that finite-difference Schr\"{o}dinger operators $\pm \BH(\om)$ have similar
qualitative spectral properties. Specifically, let $\BU(x_1,x_2) = U \one_{\{x_1=x_2\}}$,
$U = 11 = 1 + \|\BH_0\| + \|\BV\|$, $\BH_0\ge 0$, $\|\BH_0\|=8$. Then
the Dos is the same as with $U=0$, hence, contained in
$\big[0, \|\BH_0\| + \|\BV\| \,\big] = [0, 10]$
(recall that $0\le V\le 1$ in this example). On the other hand,
for $\Bx = (0,1)$, one has, by non-negativity of $\BH_0 + \BV$,
$$
\langle \one_\Bx | \BH \one_\Bx \rangle \ge U  = 11.
$$
Therefore, the spectrum of $\BH(\om)$ in $(10, 11]$ is nonempty, while the DoS vanishes
in  $(10, 11]$.

\smallskip

Nevertheless, it is true that, even though the potential \emph{random field} $\BW(\Bx;\om)$ is
\emph{not ergodic}, Hamiltonians $\BH(\om)$ still form an \emph{ergodic family of operators},
hence their spectral components are a.s. nonrandom.


\section{Basic definitions, assumptions and main results} \label{sec:basic.notions}

\subsection{Configurations of indistinguishable particle  in $\DZ^1$}
\label{ssec:interaction.potentials}

In the first part of this paper, we work with configurations of $N\ge 1$ quantum particles in
the one-dimensional lattice $\DZ$. In quantum mechanics, particles of the same kind
are considered indistinguishable; more precisely, depending on the nature of the particles,
the wave functions describing $N>1$ particles have to be either symmetric (Bose--Einstein
quantum statistics) or antisymmetric (Fermi--Dirac quantum statistics). We choose here fermionc case; this gives rise to slightly simpler notations and constructions.

For clarity, we use boldface symbols for objects related to
multi-particle systems.

Quantum states of an $N$-particle fermionic system in $\DZ$ are elements of the
Hilbert space $\bcHN = \bcH^{N,-}$ of all square-summable antisymmetric functions
$\BPsi:\DZ^N\to\DC$, with the inner product $\langle \cdot\,|\,\cdot\rangle$
inherited from the Hilbert space $\ell^{2}(\DZ^N) = \left(\cH^1\right)^{\otimes N}$.
In this particular case where the "physical" configuration space
is one-dimensional, $\bcH^N$ admits a simple representation which we will use.

First, note that
any antisymmetric function $\Bx=(x_1, \ldots, x_N) \mapsto \BPsi(x_1, \ldots, x_N)$ vanishes
on all hyperplanes $\{\Bx: x_i = x_j\}$, $1\le i < j \le N$. Further, any function
defined on the "positive sector" $\cZp=\{(x_1, \ldots, x_N): x_1>x_2 > \cdots >x_N\}$
admits a unique antisymmetric continuation to $\cZN := \DZ^N$. An orthonormal basis
in $\bcHN$ can be chosen in the usual form

\def\myotimes{\operatornamewithlimits\otimes}

$$
\BPhi_{\Ba}(\Bx) = \frac{1}{ \sqrt{N!}} \sum_{\pi\in\fS_N}
(-1)^{\pi} \myotimes_{j=1}^{N} \one_{ a_{\pi^{-1}(j)}},
\quad \Ba\in\cZN, \; \#\{a_1, \ldots, a_N\}=N.
$$
Here $\fS_N$ is the symmetric group acting in $\DZ^N$ by permutations of the particle
positions, $\pi(\Ba) = (a_{\pi^{-1}(1)}, \ldots, a_{\pi^{-1}(N)})$, and
$(-1)^{\pi}\in\{+1, -1\}$ denotes the parity of the permutation $\pi$. It is readily seen
that $\bcHN$ is unitarily isomorphic  to the Hilbert space
$\ell^2_{-}(\cZp)$ of square-summable functions on $\cZp$;
equivalently, one can consider square-summable functions on  the
set $\cZpeq = \{ (x_1 \ge \cdots \ge x_N \}$ vanishing on the boundary
$\cZpzero = \{\Bx\in\cZpeq:\, \exists\, i\ne j,\; x_i = x_j\}$. Indeed, the isomorphism
is induced by the bijection between the orthonormal bases $\{\BPhi_{\Ba}\}$
and $\{\one_{\Ba}\}$:
$$
\BPhi_{\Ba} \leftrightarrow  \one_{\Ba} =\one_{a_1} \otimes \cdots \otimes \one_{a_N},
\quad \Ba\in\cZp.
$$
The subspace $\bcHN$ is invariant under any operator commuting with the action of the symmetric
group $\fS_N$.

An advantage of the above representation is that $\cZp$ inherits its explicit, natural graph
structure from $\DZ^N$.

Occasionally we will denote by $\vBx$ the ordered configuration (i.e., a vector)
$(x_1, \ldots, x_N)$ corresponding to an unordered configuration
$\Bx=\{x_1, \ldots, x_N\}$.

\subsection{Fermionic Laplacians}
\label{ssec:interaction.potentials}

Any unordered finite or countable connected graph $(\cZ,\cE)$ with the set of vertices $\cZ$ and the set of edges $\cE$ is endowed with the canonical graph distance
$(x,y)\mapsto \rd_{\cZ}(x, y)$ (defined as the length f the shortest path
$x\rightsquigarrow y$ over the edges) and with the canonical (negative) graph Laplacian $(-\Delta_\cZ)$:
$$
(-\Delta_\cZ f)(x) = \sum_{ \nb{x,y} } \big( f(x) - f(y) \big)
= n_\cZ(x) f(x) - \sum_{ \nb{x,y} } f(y);
$$
here $\nb{x,y}$ denotes a pair of nearest neighbors, and $n_\cZ(x)$ is the coordination
number of the point $x$ in $\cZ$, i.e., the number of its nearest neighbors.

In particular, one can take $\cZ = \DZ^N$ with the edges $\nb{x,y}$ formed by
vertices $x,y$ with $|x-y|_1 := |x_1 - y_1| + \cdots + |x_N - y_N|=1$. In other words,
the vector norm $|\cdot|_1$ induces the graph distance on $\DZ^N$. The (negative) Laplacian
on $\DZ^N$, which we will now denote by $(-\BDelta)$,
can be written as follows:
$$
(-\BDelta) = \sum_{j=1}^N  \left(\myotimes_{i=1}^{j-1} \one^{(i)} \right)
\otimes \big(-\Delta^{(j)} \big)\otimes
\left( \myotimes_{k=j+1}^{N} \one^{(k)} \right)
$$
where $\one^{(i)}$ is the identity operator acting on the $i$-th variable, and
$(-\Delta^{(j)})$ is the one-dimensional negative lattice Laplacian in the $j$-th variable:
$$
(-\Delta^{(j)} f)(x_j) = 2f(x_j) - f(x_j-1) - f(x_j+1), \quad x_j\in\DZ.
$$
Its restriction to the subspace $\bcHN$ of antisymmetric functions can be equivalently defined
in terms of functions supported by the positive sector $\cZp$, hence vanishing on its border
$\cZpzero$. Indeed, the matrix elements of $\BDelta$ in the basis $\BPhi_{\Ba}$ can be nonzero only for pairs $\BPhi_\Ba$, $\BPhi_\Bb$ with $|\Ba - \Bb|_1=1$, so that,
for some $j\in[1,N]$, $|a_j - b_j|=1$, while for all $i\ne j$, $a_i = b_i$.
If $\Ba, \Bb\in\cZp$, then
$\langle \BPhi_\Ba\, | \,\BDelta \, | \,\BPhi_\Bb \rangle
= \frac{1}{N!}\langle \one_\Ba\, | \,\BDelta \, |\, \one_\Bb \rangle$.
If, say, $\Ba\in\cZpzero$,
then the respective  matrix element of the Laplacian's restriction  to $\cZp$
with Dirichlet boundary conditions on $\cZpzero$ vanishes, but so does the
function $\BPhi_\Ba$ (which is no longer an element of the basis in $\bcHN$).
Therefore, up to a constant factor, the restriction of the $N$-particle Laplacian to the fermionic subspace $\bcHN$ is unitarily equivalent to the standard graph Laplacian on $\cZp$. From this point on, we will work with the latter, occasionally making references to the space $\ell^2(\DZ^N)$.

It will be convenient to use in the course of the scaling analysis
a different notion of distance on $\DZ^N$ (hence, also on $\cZp\subset\DZ^N$):
the max-distance defined by
$$
\rho(\Bx, \By) = \max_{1\le j \le N} \rd_{\cZ}(x_j, y_j),
$$
and work with balls relative to the max-distance,
\be\label{eq:ball.factor}
\bball_L(\Bx) = \{\By:\, \rho(\Bx,\By)\le L\} = \mytimes_{j=1}^N \ball_L(x_j).
\ee
Here $\ball_L(x) = [x-L, x+L]\cap\DZ$.

In the positive sector $\cZp$, the above factorization of $\rho$-balls
is subject to the condition that the RHS  of \eqref{eq:ball.factor}
is itself a subset of the positive sector $\cZp$ (but the inclusion
"$\subset$" always holds true). Considering configurations
$\Bx=\{x_1, \ldots, x_n\}$ as subsets of $\DZ$, one can define the distance
between two configurations $\Bx'\in{\boldsymbol{\cZ^{(n')}_{>}}}$ and
$\Bx''\in{\boldsymbol{\cZ^{(n'')}_{>}}}$ in a usual way:
$$
\dist( \Bx', \Bx'') = \min_{u\in\Bx'} \min_{v\in\Bx''} |u-v|.
$$

\begin{lemma}
Let $\Bx\in\cZp$  be a union of two subconfiguraitons
$\Bx'\in{\boldsymbol{\cZ^{(n')}_{>}}}$ and
$\Bx''\in{\boldsymbol{\cZ^{(n'')}_{>}}}$, such that
$\dist(\Bx', \Bx'') > 2L$. Then the following identity holds true:
$$
\bball^{(N)}_{L}(\Bx) = \bball^{(n')}_{L}(\Bx') \times \bball^{(n'')}_{L}(\Bx'').
$$
\end{lemma}

The proof is straightforward and will be omitted.

The graph distance $\rdN$ will be useful in some geometrical
constructions and definitions, referring to the graph structure
inherited from $\cZN$.

Given a subgraph $\BLam\subset\cZp$, we define its internal, external and the
so-called graph (or edge) boundary,
in terms of the canonical graph distance:
$$
\bal
\pt^-\BLam &= \myset{ \By\in\BLam:\; \rdN(\By, \BLam^\rc) = 1},
%
\quad \pt^+\BLam
= \pt^-\BLam^\rc
\\
\pt\BLam &= \myset{ (\Bx,\By)\in\BLam\times\BLam^\rc:\; \rdN(\Bx, \By) = 1}.
\eal
$$

We also define the occupation numbers relative to a configuration
$\Bx=(x_1, \ldots, x_N)$.
Namely, define a function $\Bn_{\Bx}:\DZ \mapsto \DN$ by
$
\Bn_{\Bx}(y) = \# \{j: \, x_j = y\}, \;\; y\in\cX.
$

\subsection{Multi-particle Hamiltonians}

The matrix elements of resolvents $\BG_\BLam(E) = (\BH_\BLam - E)^{-1}$,
for $E \not\in \Sigma(\BH_\BLam)\equiv {\rm spec}(\BH_\BLam)$, in the canonical delta-basis,
usually referred to as the Green functions, will be denoted by $\BG_\BLam(\Bx,\By;E)$. In the context of random operators, the dependence upon the element $\om\in\Om$ will be often omitted, unless required or instructive. Similar notations will be used for infinite $\BLam$.

\subsection{Graph Laplacians and fermionic Hamiltonians}

Given a real-valued function $\BW:\cZp\to\DR$, we identify it with the operator of
multiplication by $\BW$ and consider the fermionic $N$-particle random Hamiltonian
$$
\BH(\om) = \BH_0 + \BW(\om),
$$
where $\BH_0$ is a second-order finite-difference operator in $\ell^2(\cZp)$, for example
the graph Laplacian $\BDelta$. We assume that
$$
\BW(\Bx) = \BW(\Bx;\om) = g\BV(\Bx;\om) + \BU(\Bx)
$$
where $g\BV(\cdot;\om)$ is the external random potential energy of the form
$$
\BV(\Bx;\om) = V(x_1;\om) + \cdots + V(x_N;\om),
$$
and $V:\DZ\times\Om\to\DR$ is a random field on $\DZ$ relative to a probability space
$(\Om,\fF,\DP)$; the expectation relative to the measure $\DP$ will be denoted by
$\esm{\cdot}$. Our assumptions on $V$ are listed below (cf. \Wone -- \Wthree).

Further, $\BU$ is the interaction energy operator; for notational brevity, we assume
that it is generated by a two-body interaction potential $U$,
$$
\BU(\Bx) = \sum_{i\ne j} U(|x_i - x_j|),
$$
satisfying one of the hypotheses \Uzero --\Uone (see below).

Note that, without loss of generality, if $\|\BH_0\|<\infty$, then one can always
assume that the absolute values of the matrix elements of $\BH_0$
are bounded by $1$. Otherwise
one can consider a two-parameter family
$\BH_0 + h^{-1}(g\BV(\Bx;\om) + \BU(\Bx)) = h^{-1}(h\BH_0 + g\BV(\Bx;\om) + \BU(\Bx))$,
so the two operators share the eigenvectors, while the matrix elements of $h\BH_0$
can be made smaller than $1$ by taking $h$ small enough.

One can also consider higher-order finite-difference operators $\BH_0$; this requires only minor
technical modifications; see the discussion in Section \ref{sec:deterministic.bounds}.

\subsection{Assumptions on the random potential}
\label{ssec:assumptions.V}

We assume that the random field $V:\DZ^d\times\Omega\to\DR$ is (possibly) correlated, but strongly mixing; this includes of course the IID potentials treated earlier in \cite{CS09a,CS09b} and  \cite{AW09a,AW09b}.
Let $F_{V,x}(t) = \pr{ V(x;\omega) \leq t}$, $x\in\DZ^d$, be the marginal probability distribution functions (PDF)  of the random field $V$, and
$F_{V,x}(t\,|\, \fF_{\ne x}) := \pr{ V(x;\om) \le t \,|\, \fF_{\ne x}}$ the  conditional distribution functions (CDF) of the random field $V$ given the sigma-algebra
$\fF_{\ne x}$   generated by random variables $\{V(y;\om), y\ne x\}$. Our assumptions on correlated potentials are summarized as follows:

\begin{enumerate}[\Wone]
  \item The marginal CDFs are uniformly H\"{o}lder-continuous: for some $\kappa>0$,
\be\label{eq:cond.continuity.V}
\essup \, \sup_{x\in\DZ^d}\, \sup_{a\in \DR}
\left[ F_{V,x}(a+s\,|\, \fF_{\ne x}) - F_{V,x}(a\,|\, \fF_{\ne x}) \right]
\le \Const\,s^\kappa.
\ee
\end{enumerate}
\begin{enumerate}[\Wtwo]
  \item (\emph{Rosenblatt mixing})
For any pair of subsets  $\ball'$, $\ball''\subset \DZ^d$ with  $\rd(\ball',\ball'')\ge L$,
any events $\cE'\in\fF_{\ball'}$, $\cE''\in\fF_{\ball''}$ and some $C>0$
\be\label{eq:Cmix}
\ba
\Big|\pr{ \cE' \cap \cE''} - \pr{ \cE'} \pr{ \cE'' } \Big|
 \le e^{-C \, \ln^2 L}.
\ea
\ee
\end{enumerate}

One can easily check that for any $p>0$, $\alpha\in(1,2)$,  $b\in(0,1)$, arbitrarily large $a>0$ and sufficiently large $L_0\in\DN$
\be\label{eq:Cmix.prob}
e^{-C \ln^2 L_0^{\alpha^k}} < \left( L_0^{\alpha^k} \right)^{-ap(1+b)^k}, \quad k\ge 0.
\ee
The rate of decay of correlations indicated in the RHS of \eqref{eq:Cmix.prob} is required to prove dynamical localization bounds with decay rate of EF correlators faster than polynomial, and it can be relaxed to a power-law decay, if one aims to prove only a power-law decay of EF correlators\footnote{In this case, \Wone\, can be relaxed to a form of log-H\"{o}lder continuity of the marginal CDFs.}.

Assumptions \Wone--\Wtwo\, are sufficient for the proof of spectral and strong dynamical localization in multi-particle systems. In particular, the role of \Wone\, is to guarantee Wegner-type estimates used in the MPMSA scheme. However, it was discovered in \cite{AW09a,AW09b} and in \cite{CS09a,CS09b} that traditional, Wegner-type EVC estimates do not provide all necessary information for efficient decay bounds on the eigenfunctions of multi-particle operators. More precisely, conventional EVC bounds seem so far insufficient for the proof of the exponential decay of eigenfunctions with respect to a norm in the configuration space of $N$-particle systems, starting from\footnote{For $N=2$, one can use a simpler EVC bound proven in \cite{CS09a}.}  $N=3$. For this reason, we proposed earlier \cite{C10} a new method allowing to compare spectra of two strongly correlated multi-particle subsystems and extending the  Wegner-type EVC estimate. For the new method to apply, one needs an additional assumption on the random potential field which we will describe now.

First, introduce the following notations. Given a lattice parallelepiped  $Q\subset \DZ^d$, we denote by $\xi_{Q}(\omega)$ the sample mean of the random field $V$ over the $Q$,
$$
 \xi_{Q}(\omega) =  \langle V \rangle_Q  = | Q |^{-1} \sum_{x\in Q} V(x,\omega)
$$
and define the "fluctuations" of $V$ relative to the sample mean,
$
 \eta_x  = V(x,\omega) - \xi_{Q}(\omega), \; x\in Q.
$
Denote by $\fF_{V, Q}$ the sigma-algebra generated by $\{\eta_x, V_y:\, \,x\in Q, y\not\in Q\}$, and by
$F_\xi( \cdot\,| \fF_{V, Q})$ the conditional distribution function of $\xi_Q$ given
$\fF_{V, Q}$.
Further, consider parallelepipeds $Q\subset\DZ^d$ with $\diam(Q) \le R<\infty$, and introduce an
$\fF_{V, Q}$-measurable random variable
\begin{equation}\label{eq:def.nu.R}
\nu_R(s;\om) :=
\sup_{t\in\DR}    |F_\xi(t+s\,| \fF_{V, Q}) - F_\xi(t\,| \fF_{V, Q})|,
\; s\ge 0.
\end{equation}

We will assume that the random field $V$ fulfills the following condition:
\par\medskip
\Wthree:
\textit{ There exist $C',C'', A', A'', b', b''\in(0,+\infty)$ such that
\begin{equation}\label{eq:CMnu}
\forall \, s\in[0,1]\;\;
{\rm ess} \sup
\pr{ \nu_R(s;\om) \ge C' R^{A'} s^{b'}} \le C'' R^{A''} s^{b''}.
\end{equation}
} 

In the particular case of a Gaussian  IID field, e.g.,  with zero mean and unit variance, $\xi_Q$ is a Gaussian random variable with variance ${|Q|}^{-1}$ independent of the "fluctuations" $\eta_x$, so that its probability density $p_{V,|Q|}$ is bounded,
although $\| p_{V,|Q|}\|_\infty \sim |Q|^{1/2}\to\infty$ as $|Q|\to\infty$.
The property \Wthree\, has been recently proven for a larger class of IID potentials by Gaume \cite{G10} in the framework of his PhD project.
Note that the property \Wthree\, is not quite obvious, since the conditional distribution given by $F_\xi(t\,| \fF_{V, Q})$ can be very singular -- even discontinuous -- for \emph{some} conditions, e.g., in the case of an IID random field $V$ with uniform marginal distribution.
Nevertheless, the proof of \Wthree for an IID random field with uniform marginal probability distribution  is quite simple.

\subsection{Assumptions on the interaction potential}
\label{ssec:assumptions.U}

We assume that the interaction potential $U$ generating the interaction
energy $\BU$
satisfies one of the following decay conditions:
\begin{enumerate}[\Uzero]
  \item \emph{There exists $r_0<\infty$ such that
\be\label{eq:U.decay}
\forall\, r\ge r_0\quad  U(r) = 0.
\ee
}
\end{enumerate}

%

\begin{enumerate}[\Uone]
  \item \emph{There are some $C\in(0,+\infty)$, $\delta\in\big(0,\frac{1}{14}\big)$ and
  $\theta\in\big(0, \frac{\delta}{1+\delta}\big)$ such that
\be\label{eq:U.decay}
\forall\, r\ge 0 \quad | U(r)| \le C e^{-c r^{1-\theta}}.
\ee
}
\end{enumerate}

Naturally, \Uzero $\Rightarrow$ \Uone.
We will prove the multi-particle localization first under the strongest
assumption \Uzero (leading to a simpler proof), to illustrate the general structure
of the new MPMSA scheme, and then
extend the proof to interactions satisfying \Uone.

\subsection{Main results}

\begin{theorem}\label{thm:Main.DL}
Assume that the random field $V$ fulfills conditions \Wone--\Wthree, and the interaction $\BU$ fulfills one of the conditions \Uzero, \Uone. There exists $g_0\in(0,+\infty)$ such that if $|g|\ge g_0$, then with probability one, the fermionic operator $\BH(\om)=\BDelta + g\BV(\om)+\BU$ has pure point spectrum and all its eigenfunctions $\BPsi_j(\om)$ are exponentially decaying at infinity:  for each $\BPsi_j$, some $\hBx_j$  and all $\Bx$ with $\rho(\hBx,\Bx)$ large enough,
\be\label{eq:thm.Main.DL.1}
| \BPsi_j(\Bx,\om) | \le e^{-m \rho(\hBx,\Bx)}, \; m>0.
\ee
In addition, for all points $\Bx,\By$ with $\rho(\Bx,\By)$ large enough and some $a,c>0$,
for any bounded Borel function $f:\DR\to\DR$
\be\label{eq:thm.Main.DL.2}
\esm{  \big| \langle \one_{\By}\, | \,  f(\BH(\omega)) \, | \, \one_{\Bx}\rangle\big| }
\le e^{ -a\ln^{1+c} \rho(\Bx,\By) } \|f\|_\infty.
\ee
Consequently,  for any finite subset $\BK\subset\DZ^{Nd}$
\be\label{eq:thm.Main.DL.3}
\esm{ \Big\| e^{ a\ln^{1+c} \BX} f(\BH(\omega))  \one_{\BK} \Big\|} < C(\BK) \, \|f\|_\infty
<\infty
\ee
where operator
$\BX_{\BK}$ is the defined by
$(\BX_{\BK} f)(\Bx) = (\rho(\BK,\Bx)+1) f(\Bx)$.
In particular,
\be\label{eq:thm.Main.DL.4}
\esm{ \sup_{t\in\DR} \, \Big\|e^{ a\ln^{1+c} \BX} e^{ -it \BH(\omega)} \one_{\BK} \Big\|} < \infty.
\ee
\end{theorem}

\section{Deterministic bounds}
\label{sec:deterministic.bounds}


\subsection{Geometrical resolvent inequality}

The most essential part of the scaling analysis concerns finite-volume approximations
$\BH_\BLam(\om)$ of the random Hamiltonian $\BH(\om)$, acting in finite-dimensional spaces $\bcH_\BLam:=\ell^2(\BLam)$,
$\BLam\subset\cZN$, $|\Lam|\equiv \card\Lam < \infty$.

Operator $\BDelta$ can be represented as follows:
$$
-\BDelta = n_{\cZN}  - \sum_{ \nb{\Bx,\By} } (\Gamma_{\Bx \By} + \Gamma_{\By \Bx}),
\quad
(\Gamma_{\Bx \By} f)(\Bx) := \delta_{\Bx \By} f(\By),
$$
and $\delta_{\Bx \By}$ is the Kronecker symbol. Similar formulae are valid for the restriction $\Delta_\Lam$ of $\Delta$ to a finite subset $\BLam\subset\cZN$ with Dirichlet boundary conditions outside $\Lam$; in this case, one has to keep only the pairs $\nb{\Bx,\By}\in \BLam\times\BLam$.

We denote by $\BG_{\BLam}(E) = (\BH_{\BLam}  - E)^{-1}$ the resolvent of $\BH_{\BLam}$ and by $\BG(\Bx,\By;E)$ the matrix elements thereof (a.k.a. Green functions) in the standard delta-basis.

The so-called Geometric Resolvent Inequality for the Green functions can be easily deduced from the second resolvent identity:
\be\label{eq:GRI.equal}
\left| \BG_{\bball_L}(\Bx,\By;E) \right|
\le C_\ell \,  \max_{\Bv \in\pt^- \bball_\ell(\Bx)} |\BG_{\bball_\ell(\Bx)}(\Bx,\Bv;E) |
\; \max_{\Bv'\in\pt^{+} \bball_\ell(\Bx)} \, \left|\BG_{\bball_L}(\Bv', \By; E) \right|
\ee
where
\be\label{eq:GRI.C.d}
C_\ell = |\pt  \bball_\ell(\Bv)| \le C(N,d) \ell^{Nd}.
\ee
Clearly, \eqref{eq:GRI.equal} implies the  inequality
\be\label{eq:GRI}
\left| \BG_{\ball_L}(\Bx,\By;E) \right|
\le \Big(\, C_\ell \,  \max_{\rho(\Bx,\Bv)=\ell} |\BG_{\bball_\ell(\Bx)}(\Bx,\Bv;E) | \,\Big)
\; \max_{\rho(\Bu,\Bv')\le \ell+1} \, \left|\BG_{\bball_L}(\Bv', \By; E) \right|
\ee
which is weaker than \eqref{eq:GRI.equal}, but sufficient for the purposes of the scaling analysis. Sometimes we shall use the following inequality which stems from
\eqref{eq:GRI.equal}:
\be\label{eq:GRI.equal.S.NR}
\left| \BG_{\bball_L}(\Bx,\By;E) \right|
\le C_\ell \,  \| \BG_{\bball_\ell(\Bu)}(E) \|
\; \max_{\Bv: \rho(\Bu,\Bv) \le \ell+1} \, \left|\BG_{\bball_L}(\Bv, \By; E) \right|.
\ee
Similarly, for the solutions $\psi$ of the eigenfunction equation $\BH\Bpsi = E\Bpsi$ we have
\be\label{eq:GRI.EF}
| \Bpsi(\Bx) |
\le \, C_\ell \, \,  \| \BG_{\bball_\ell(\Bx)}(E) \| \,
\; \max_{\By: \rho(\Bx,\By)\le \ell+1} \, | \Bpsi(\By) |,
\ee
provided that $E$ is not an eigenvalue of the operator $\BH_{\bball_\ell(\Bu)}$ and
$\Bx\in\bball_\ell(\Bu)$.

\subsection{Radial descent bounds for Green functions and eigenfunctions }\label{ssec:descent}

The analytic statements of this subsection apply indifferently to any LSO, regardless of their single- or multi-particle structure of the potential energy operator. To avoid any confusion, we denote the dimension of the lattice by $D$; in applications to $N$-particle models in $\DZ^d$, one has to set $D=Nd$.


\begin{definition}\label{def:SubH.IID}
fix an integer $\ell\geq 0$ and a  number $q\in(0,1)$. Consider a finite connected subgraph $\Lam\subset\DZ^D$  and a ball $\ball_R(u)\subsetneq\Lam$. A function
$f:\, \ball_R(u)\to\DR_+$ is called $(\ell,q)$-subharmonic in $\ball_R(u)$ if for any
ball $\ball_\ell(x)\subset\Lam_R(u)$ one has
\be\label{eq:def.l.q.subh.IID}
f(x) \leq q \;\;\mymax{y:\, \rd(x,y) \le\ell} f(y).
\ee
\end{definition}

\begin{lemma}\label{lem:SubH.IID}
Le be given integers $L\ge \ell\ge 0$ and a number $q\in(0,1)$.
If $f: \Lam\to\DR_+$ on a finite connected graph $\Lam$ is $(\ell,q)$-subharmonic
in a ball $\ball_L(x)\subsetneq\Lam$, then
\be\label{eq:lem.SubH.IID}
 f(x) \leq q^{ \left\lfloor \frac{L+1}{\ell+1} \right\rfloor} \cM(f, \ball)
 \leq q^{ \frac{L - \ell}{\ell+1} } \cM(f, \ball),
 \quad\cM(f,  \Lam) := \max_{x\in  \Lam} |f(x)|.
\ee
\end{lemma}

\proof The claim follows from \eqref{eq:def.l.q.subh.IID} by induction
(``radial descent''):
$$
f(x) \le q \cM(f, \ball_{\ell+1}) \le
\cdots \le  q^j \cM(f, \ball_{j(\ell+1)})
\le \cdots \le q^{ \left\lfloor \frac{L+1}{\ell+1} \right\rfloor} \cM(f,\ball).
\qed
$$

\begin{lemma}\label{lem:BiSubH.IID}
Let $\Lam$ be a finite connected graph and $f:\Lam\times\Lam\to \DR_+$,
$(x',x'')\mapsto f(x',x'')$, a function
which is separately $(\ell,q)$-subharmonic in $x'\in\ball_{r'}(u')\subsetneq\Lam$
and in $x''\in\ball_{r''}(u'')\subsetneq\Lam$, with $\rd(u',u'')\ge r'+r''+2$.
Then
\be\label{eq:RDL.IID}
 f(u', u'') \leq q^{\left\lfloor \frac{r'+1}{\ell+1} \right\rfloor
    + \left\lfloor \frac{r''+1}{\ell+1} \right\rfloor} \cM(f, \ball)
    \leq q^{ \frac{r'+r'' - 2\ell}{\ell+1} } \cM(f, \ball).
\ee
\end{lemma}
\proof
Fix any point $y''\in\ball_{r''+1}(u'')$ and consider the function
$$
f_{y''}: y' \mapsto f(y', y'')
$$
which is $(\ell,q)$-subharmonic in $\ball_{r'}(u')\subsetneq \Lam$, so by Lemma
\ref{lem:SubH.IID},
\be\label{eq:lem.BiSH.1}
f_{y''}(u') =
f(u',y'') \le q^{ \left\lfloor \frac{r'+1}{\ell+1} \right\rfloor} \cM(f, \Lam\times\Lam).
\ee
Now introduce the function $\tf_{u'}:y''\mapsto f(u',y'')$ which is
 $(\ell,q)$-subharmonic in $\ball_{r''}(u'')\subsetneq \Lam$ and  bounded by  the RHS of \eqref{eq:lem.BiSH.1}. Applying again Lemma \ref{lem:SubH.IID}, the claim follows.
\qedhere

The relevance of the above notions and results is illustrated by the following

\begin{lemma}[Cf. Lemma 3.3 in \cite{C11}]\label{lem:cond.SubH.IID}
Consider a ball $\ball_L(u)$ and an operator $H=H_{\ball_L(u)}$ with fixed (non-random) potential $V$. Let $\{\psi_j, \, j=1, \ldots,  |\ball_L(u)|\}$ be the normalized eigenfunctions of $H$. Pick a pair of points $x',y' \in \ball_L(u)$ with $\rd(x',y') > 2(\ell + 1)$ and an integer
$R \in[\ell+2, \, \rd(x',y') - (\ell+2)]$. Suppose that any ball $\ball_{\ell}(v)$ with $v\in\ball_R(x')$ is \EmNS, and set
\be\label{eq:lem:ConSubH.IID}
q =  q(D,\ell; E) =C(D)  e^{-\gamma(m,\ell)\ell}
\ee
with the constant $C(D)$ defined in the same way as $C(d)$ in \eqref{eq:GRI.C.d}. Then:
\begin{enumerate}[{\rm(A)}]
  \item the kernel $\Pi_{\psi_j}(x,y)$ of the spectral projection
  $\Pi_{\psi_j} = \ketbra{\psi_j}{\psi_j}$ is $(\ell+1,q)$-subharmonic in   $x \in \ball_R(x')$, with global maximum $\le 1$;
  \item if $\ball_L(u)$ is also \NR, then the Green functions $G(x,y;E)=(H-E)^{-1}(x,y)$ are $(\ell+1,q)$-subharmonic in $x \in \ball_R(x')$, with global maximum $\le e^{L^\beta}$.
\end{enumerate}
\end{lemma}

\begin{remark}
The general idea of "two-sided" bounds on the Green functions (as functions of \emph{two} arguments) can be found already in the Spencer's paper \cite{Spe88}. This allows to avoid a more
complex geometrical construction going back to \cite{DK89}, required when more
than one "unwanted", "singular" locus is to be \emph{allowed} in the course of the scaling step.
\end{remark}

\begin{remark}
An extension of Lemma \ref{lem:cond.SubH.IID} to finite-difference kinetic energy operators
$\BH_0$ of higher order is quite straightforward. Specifically, if $\BH_0$ has range
$\ell_0\ge 1$, i.e., $\langle \one_{\Bx}| \BH_0 \one_{\By}\rangle = 0$ for $\rho(\Bx,\By)>\ell_0$,
the denominators $\ell+1$ in Eqns \eqref{eq:lem.SubH.IID}, \eqref{eq:RDL.IID}, \eqref{eq:lem.BiSH.1} are to be replaced by $\ell+\ell_0$. Recall that we assume (implicitly)
that the matric elements of $\BH_0$ are bounded by $1$ (as they are for the graph Laplacian),
but the reduction to this case can be made by rescaling the operator $\BH_0$, i.e.,
by replacing $\BH_0 \rightsquigarrow h\BH_0$, with $h>0$ small enough (provided that $\|\BH_0\|<\infty$). Apart from a modification of some auxiliary constants, this has no impact on
the qualitative final result on multi-particle localization under the assumption of strong disorder.
\end{remark}

\subsection{Localization and tunneling in finite balls}
\label{ssec:tun.and.loc.finite.balls}

\begin{definition}\label{def:CNR.IID}
Given a sample $V(\cdot;\om)$, a ball $\bball_L(\Bu)$ is called
\begin{enumerate}[\;\;$\bullet$]
  \item $E$-non-resonant (\NR\,) if $\|\BG_\bball(E;\om)\|\leq e^{+L^{\beta} }$,
  and $E$-resonant
  otherwise;
  \item completely $E$-non-resonant ($E$-CNR) if it does not contain any \ER\, ball $\bball_\ell(\Bu)\subseteq \bball_L(\Bu)$ with $\ell\ge L^{1/\alpha}$, and $E$-partially resonant (\PR), otherwise.
\end{enumerate}
\end{definition}

\begin{definition}\label{def:S}
Given a sample $V(\cdot;\om)$, a ball $\bball_L(\Bu)$ is called $(E,m)$-non-singular (\EmNS) if
\be\label{eq:def.NS}
 \mymax{y\in\pt  \bball_L(\Bu)} \; |G(\Bu,\By;E;\om)| \leq e^{-\gamma(m,L) L + 2L^\beta},
\ee
where
\be\label{eq:def.gamma}
 \gamma(m,L)  :=  m\left(1 + L^{-\tau}\right),
 \qquad \tau = 1/8.
\ee
Otherwise, it is called $(E,m)$-singular (\EmS).
\end{definition}

In the next definition, we use a parameter $\varrho = (\alpha - 1)/2$; with $\alpha = 4/3$,
one obtains $\varrho=1/6$ and $(1+\varrho)/\alpha = 7/8$.

\begin{definition}\label{def:loc}
Given a sample $V(\cdot;\om)$,  a ball $ \bball_L(\Bu)$ is called $m$-localized (\loc, in short) if any eigenfunction $\Bpsi_j$ of the operator $\BH_{ \bball_L(\Bu)}$ operator satisfies
\be\label{eq:def.mloc}
|\Bpsi_j(\Bx) \, \Bpsi_j(\By)| \le e^{-\gamma(m,L) \rho(\Bx,\By)}
\ee
for any pair of points $\Bx,\By\in  \bball_L(\Bu)$ with
$\rho(\Bx,\By) \ge L^{\frac{1+\varrho}{\alpha}} \equiv L^{7/8}$.
\end{definition}

\begin{definition}\label{def:distant}
We will call a pair of balls $\bball_L(\Bx)$, $\bball_L(\By)$ \emph{distant} iff
$\rho(\Bx,\By)\ge C_N L$, with $C_N = 2A_N+3$, $A_N=4N$.
\end{definition}

(The role of the constants $A_N$ and $C_N$ will become clear
in subsection \ref{ssec:PI.balls}.)

\begin{definition}\label{def:tun}
A ball $\bball_{\ell^{\alpha}}(\Bu)$ is called $m$-tunneling (\mT) if it contains a pair of
distant \nloc\, balls  $\bball_{\ell}(\Bx)$, $\bball_{\ell}(\By)$, and
$m$-non-tunneling (\NT), otherwise.
\end{definition}

It is to be emphasized that, unlike the property of $E$-resonance or $(E,m)$-singularity, the tunneling property \textbf{is not related to a specific value of energy $E$}, and  tunneling even in \textit{a single} ball occurs with a small probability.

\begin{lemma}\label{lem:loc.and.NR.imply.NS.IID}
Let $E\in\DR$ and an \loc\, ball $ \bball_L(\Bu)$ be given.
If $ \bball_L(\Bu)$ is also  \NR, then it is \EmNS.
\end{lemma}
\proof
The matrix elements of the resolvent $\BG_\bball(E)$ can be assessed as follows:
\be\label{eq:G.basis}
|\BG_\bball(\Bx,\By;E)| \le \sum_{E_j\in\sigma(H_{\bball_L(\Bu)})}
 \frac{ |\Bpsi_j(\Bx)| \, |\Bpsi_j(\By)| }{|E - E_j| }\,.
\ee
If $\dist(E, \sigma(H_\bball)) \ge e^{-L^\beta}$ and
$\ln (2L+1)^d \le L^\beta$, then the \loc \, property implies
$$
\bal
\displaystyle |\BG_\bball(\Bx,\By;E)|  &
& \le e^{-\gamma(m,L)L + L^\beta + \ln |\bball|}
 \le e^{-\gamma(m,L)L + 2L^\beta}. \qed
 \eal
$$
Observe that with $m\ge 1$ and $1-\tau > \beta$, for $L_0$ large enough
\be\label{eq:NS.exp}
m(1 + L_k^{-\tau})L_k - 2L_k^\beta = mL_k + mL_k^{1-\tau} - 2L_k^\beta
\ge m(1 + {\textstyle\half}  L_k^{-\tau}) L_k.
\ee
We will see that the condition $m\ge 1$ is fulfilled for $|g|$ large enough; cf. Lemma \ref{lem:L.0.Np}.

From this point on, we will work with a sequence of "scales" - positive integers $\{L_k, k\ge 0\}$ defined recursively by $L_{k+1} = \lceil L_k^{\alpha} \rceil$, $L_0>2$.
For clarity, we keep the value $\alpha=4/3$; observe that $\alpha^2<2$. In several arguments we will require the initial scale $L_0$ to be large enough.


\begin{lemma}\label{lem:CNR.and.no.2.S.imply.NS}
There is $\tilde L^{(1)}\in\DN$ such that for  $L_0\ge \tilde L^{(1)}$

\begin{enumerate}[\rm (A)]
  \item if a ball $ \bball_{L_{k+1}}(\Bu)$ is \CNR\, and contains no pair of
  distant \EmS\, balls $\bball_{L_{k}}(\Bv)$, $\bball_{L_{k}}(\Bw)$,
  then it is also \EmNS;
   \item if a ball $ \bball_{L_{k+1}}(\Bu)$ is \NT and \CNR, then it is also \EmNS.
\end{enumerate}

\end{lemma}
\begin{proof}
\noindent
(A) By assumption, either $ \bball_{L_{k+1}}(\Bu)$ is \CNR\, and contains no \EmS\, ball
of radius $L_{k}$, or there is a point $\Bw \in\bball_{L_{k+1}}(\Bu)$ such that
any ball $\bball_{L_{k}}(\Bv)\subset\bball_{L_k}(\Bu)$ with
$\rho(\Bw,\Bv)\ge C_N L_{k}$ is \EmNS
\footnote{The value of $C_N$ is stipulated in Definition \ref{def:distant}, but it is
irrelevant for this proof, provided that $L_0$ is large enough.}.

In the former case, such an exclusion is unnecessary, but in order to treat
both situations with one argument, we can formally set $\Bw = \Bu$ (or any other point).

Fix points $\Bx,\By$ with
$R:=\rho(\Bx,\By)>L_{k}^{1+\varrho}$.
By triangle inequality,
$$
\textstyle
\rho(\Bx, \bball_{(C_N - 1) L_{k}}(\Bw))  +
\rho(\By, \bball_{ (C_N - 1) L_{k}}(\Bw))
\ge R - (2C_N - 2) L_{k}.
$$
Assume first that
$$
\bal
r'  &:= \rho(\Bx, \bball_{ (C_N-1) L_{k}}(\Bw)) \ge  L_{k} +1,
\\
r'' &:= \rho(\By, \bball_{(C_N - 1) L_{k}}(\Bw)) \ge  L_{k} +1.
\eal
$$
All balls  of radius $L_{k}$ both in $\bball_{r'}(\Bx)$ and in $\bball_{r''}(\By)$ are
automatically disjoint from $\bball_{2L_{k}}(\Bw)$, thus \EmNS.
Furthermore,
$$
r' + r'' \ge R - 2(C_N - 1)L_{k} - 2 \ge  R - 2C_NL_{k}.
$$

Consider the set
$\bball =  \bball_{r'}(\Bx)  \times \bball_{r''}(\By)$ and the function  $f: \bball \to \DC$ defined by
$f(\Bx', \Bx'') = \BG_{\bball_{L_{k+1}}}(\Bx', \Bx''; E)$.
Since $E$ is not a pole of the resolvent $\BG_{\bball_{L_{k+1}}(\Bu)}(\cdot)$, it is well-defined
(hence, bounded, on a finite set).
By Lemma \ref{lem:cond.SubH.IID}, $f$ is  $(L_{k}, q)$-subharmonic both in $\Bx'$ and in $\Bx''$, with
$q \le e^{-\gamma(m, L_{k},n)}$.
Therefore, one can write, with the convention $-\ln 0 = +\infty$, using Lemma \ref{lem:BiSubH.IID} and setting for brevity $K := M - n+1$:
\be\label{eq:proof.lem.CNR.and.no.2.S.imply.NS}
\bal
-\ln f(\Bu, \By)
&\ge
-\ln \Bigg[ \left(e^{ -m(1+ \half L_{k}^{-\tau})L_{k} }
\right)^{\frac{R - 2C_N L_{k}- 2 L_{k}  }{L_{k}+1}}
e^{L_{k+1}^\beta}
\Bigg]
\\
& = m\left(1+ {\textstyle \half } L_{k}^{-\tau} \right)
    {\textstyle \frac{L_{k}}{ L_{k} + 1} } R\left(1 - 3C_N R^{-1} L_{k} \right)
    - L_{k+1}^\beta
\\
&
\textstyle
 = mR \left[ \left(1+ {\textstyle \half } L_{k}^{-\tau} \right)
   (1 -  L^{-1}_{k})  \left(1 - 3C_N L_{k}^{-\varrho } \right)
    - \frac{L_{k+1}^\beta}{mR} \right]
\\
&\ge m\left(1+ {\textstyle \quart} L_{k}^{-\tau} \right) R
\ge \gamma(m,L_{k+1}) R.
\eal
\ee
If $r'=0$ (resp., $r''=0$), the required bound follows from the subharmonicity of the function $f(\Bx',\Bx'')$ in $\Bx''$ (resp., in $\Bx'$).

\par\smallskip\noindent
(B) Assume otherwise. Then, by assertion (A), the \CNR\, ball $\bball_{L_{k+1}}(\Bu)$ must contain a pair of distant \EmS\,  balls  $\bball_{L_{k}}(\Bx)$, $\bball_{L_{k}}(\By)$. Both of them are \NR, since  $\bball_{L_{k+1}}(\Bu)$ is  \CNR. By virtue of Lemma  \ref{lem:loc.and.NR.imply.NS.IID}, both $\bball_{L_{k}}(\Bx)$ and
$\bball_{L_{k}}(\By)$ must be \nloc, so that $\bball_{L_{k+1}}(\Bu)$ must be \mT,
which contradicts the hypothesis.
\end{proof}

\begin{lemma}\label{lem:NT.implies.nloc.IID}
There is $\tilde L^{(2)}\in\DN$ such that, for all  $L_0\ge \tilde L^{(2)}$, if for any $E\in\DR$ a ball $\bball_{L_{k+1}}(\Bu)$ contains no pair of distant \EmS\,  balls  of radius $L_{k}$,
then it is \loc.
\end{lemma}
\proof

One can proceed as in the proof of the previous lemma, but with the functions
$
 f_j:(\Bv', \Bv'') \mapsto  |\Bpsi_j(\Bv') \Bpsi_j(\Bv'')|,
$
where $\Bpsi_j, j\in[1, |\bball_{L_{k+1}}(\Bu)|]$, are normalized eigenfunctions of operator
$\BH_{\bball_{L_{k+1}}(\Bu)}$. Notice that the $E_j$-non-resonance condition
for $\bball_{L_{k+1}}(\Bu)$
is not required here, since $\|\Bpsi_j\|=1$, so the function $|f_j|$ is globally bounded by  $1$.
Let $\Bx,\By\in \bball_{L_{k+1}}(\Bu)$ and assume that $\rho(\Bx,\By)=R \ge L_{k}^{1+\varrho}$.

Arguing as in the previous lemma and using the subharmonicity of the function
 $f_j(\cdot\,,\cdot)$ (cf. Lemma \ref{lem:cond.SubH.IID}), we obtain, for any pair
 of points with $\rho(\By',\By'')\ge L_{k}^{1+\varrho}$:
\be\label{eq:proof.lem.NT.implies.nloc.IID}
\bal
-\ln |\Bpsi_j(\Bx) \Bpsi_j(\By)| &\ge
-\ln \Bigg[
\left(
e^{ -m(1+ \half L_{k}^{-\tau})L_{k} }  \right)^{\frac{R - 2C_N L_{k}-2L_{k} }{L_{k}+1}}
\Bigg].
\eal
\ee
A direct comparison with the RHS of the first equation in
\eqref{eq:proof.lem.CNR.and.no.2.S.imply.NS} shows that the RHS of
\eqref{eq:proof.lem.NT.implies.nloc.IID} is bigger, owing to the absence of
the factor $e^{L_{k}^\beta}$. Therefore, it admits the same (or better) lower bound
as in in \eqref{eq:proof.lem.CNR.and.no.2.S.imply.NS}.
\qedhere

Below we summarize the relations between the parameters used in the  above proof:
\be\label{eq:relat.param.induct}
\rho(\Bx, \By) \ge L_{k-1}^{1+\varrho},
\; 0< \tau < \rho, \; 1+\varrho < \alpha; \; \beta < 1 -\tau.
\ee
These conditions are fulfilled, e.g., for
\be\label{eq:relat.param.induct.fixed}
\alpha = 4/3,\; \rho = 1/6, \; \tau = 1/8, \; \beta = 1/2.
\ee

\section{Scaling analysis of eigenfunctions. Finite-range interaction}
\label{sec:MPMSA.finite.range}

\subsection{Initial scale estimates}

\begin{lemma}[Initial scale bound]\label{lem:L.0.Np}
For any $L_0>2$, $m\ge 1$ and $p>0$ there is $g_0=g_0(m,p,L_0) < \infty$ such that for all $g$ with
$|g|\ge g_0$, any ball $\bball_{L_{0}}(u)$ and any $E\in\DR$
\begin{eqnarray}
\pr{ \bball_{L_{0}}(\Bu) \text{ is } (E,m)\text{{\rm-S}} } &\le L_{0}^{-p}
\label{eq:L0.GF.Np} \\
\pr{ \bball_{L_{0}}(\Bu) \text{ is } \mathnloc } &\le L_{0}^{-p}.
\label{eq:L0.loc.Np}
\end{eqnarray}

\end{lemma}

See the proof in Appendix.

\subsection{Eigenvalue concentration bounds for distant pairs of balls }
\label{ssec:EVC.PCT}

\begin{theorem}[Cf. \cite{C10}]
\label{thm:PCT}
Let $V: \cZ\times \Omega \to \DR$ be a random field satisfying \Wthree. For all $L>0$ large enough,
any pair of balls $\bball^{(N)}_{L}(\Bu')$, $\bball^{(N)}_{L}(\Bu'')$
with $\rdS(\Bu', \Bu'') > 4NL$
%
the following bound holds for any $s\in(0,1]$:
\begin{equation}\label{eq:main.W2.bound}
\pr{ \dist(\sigma(\BH_{\bball_{L}(\Bu')}), \sigma(\BH_{\bball_{L}(\Bu'')})) \le s }
     \le  (2L+1)^{2Nd}\, h_L (2s)
\end{equation}
where
\be\label{eq:h.L}
h_L(s) := C''L^{A''} (2s)^{b''}
+  |\bball_{L}(\Bu')| \, |\bball_{L}(\Bu'')|\, C'L^{A'} (2s)^{b'}.
\ee
Consequently, this implies the following bound
for $\beta'\in(0,\beta)$, all $L$ large enough and any pair
$\Bu'$, $\Bu''$, with $\rd(\Bu', \Bu'') > 4NL$:
\be\label{eq:cor.PCT}
\pr{ \exists\, E\in \DR:\, \text{ $\bball_{L}(\Bu')$ and
$\bball_{L}(\Bu'')$ are \PR} } \le e^{-L^{\beta'} }
\ee
\end{theorem}

See the proof in Appendix.

\subsection{"Partially interactive" $N$-particle balls }
\label{ssec:PI.balls}

Unlike the situation in single-particle models where the configuration space (e.g., the lattice $\DZ^d$) is "homogeneous", the presence of a non-trivial interaction singles out the loci
$(x_1, \ldots, x_N)\in\left(\DZ^{d}\right)^N$ where the interaction cannot be neglected. In the physically relevant case of interactions rapidly decaying at infinity, the "diagonal" subsets like
$$
\DD^{(j,k)} := \{ \Bx=(x_1, \ldots, x_N): \, x_j=x_k  \}, \quad 1\le j\ne k \le N,
$$
and particularly the "principal diagonal"
$$
\DD := \{ \Bx=(x, x, \ldots, x): \, x\in\DZ^d  \},
$$
along with their neighborhoods, are the zones where the inter-particle interaction is concentrated. Respectively, far away from all "diagonals" the interaction is either zero (if the range of interaction is finite) or weak, thus negligible in the framework of some perturbative analysis. It turns out that a tedious geometric combinatorial analysis of the "diagonal" zones is not necessary for establishing  qualitative results of the multi-particle Anderson localization theory. As was illustrated in \cite{CS09b}, a much simpler classification of balls $\bball_L(\Bu)$ suffices here:
\begin{itemize}
  \item balls close to the "principal diagonal" $\DD$\,;
  \item balls sufficiently distant from $\DD$.
\end{itemize}

In the space $\cZN$ of indistinguishable configurations $\Bx$ of $N$ particles, occupying pairwise
distinct positions $x_1, \ldots, x_N\in\cZ$, we can no longer refer directly  to the "diagonal"
$\DD$. However, it is not the proximity to this set which is important per se for our analysis,
but an easily verifiable fact that, for a point $\vBx=(x_1, \ldots, x_N)\in\cZ^N$
$$
\dist(\vBx, \DD) \le \diam \{x_1, \ldots, x_N\} \le 2 \dist(\vBx, \DD).
$$
The quantity $\diam \Bx=\diam \{x_1, \ldots, x_N\}$ is non-ambiguously defined for a configuration
$\Bx = \{x_1, \ldots, x_N\}$ of indistinguishable particles.

\begin{definition}
An $N$-particle ball $\bball_L(\Bu)$ is said partially interactive (PI, in short) if
$\diam \Bu > A_N L$, $A_N = 4N$,
and fully interactive (FI), otherwise.
\end{definition}

The local Hamiltonian  $\BH^{(N)}_{\bball_L(\Bu)}$ with a finite-range interaction in a PI ball can is algebraically decomposable in the following way:
\be\label{eq:Ham.decomposable}
\BH^{(N)}_{\bball_L(\Bu)} = \BH^{(n')}_{\bball_L(\Bu')}\otimes \one^{(n'')}
+ \one^{(n')} \otimes \, \BH^{(n'')}_{\bball_L(\Bu'')}
\ee
with $n'+n''=N$, $\Bu=(\Bu',\Bu'')$, $\Bu'\in\DZ^{dn'}$, $\Bu''\in\DZ^{dn''}$, while for a fully interactive ball one cannot guarantee such a property.
(In the case of a rapidly decaying interaction, the RHS is close in norm to the LHS.)

\begin{lemma}\label{lem:FI.diam}
\begin{enumerate}[\rm (A)]
  \item If $\bball_L(\Bx)$ is {\rm PI}, then there exist non-empty complementary subsets
  $\cJ,\cJ^c\subset[1,N]$ such that
  $\rho(\Pi_{\cJ} \bball_{L}(\Bx), \Pi_{\cJ^c} \bball_{L}(\Bx) ) > 2L$.

  \item If $\bball_L(\Bx)$, $\bball_L(\By)$ are two {\rm FI} balls with
  $\rho(\Bx,\By)> C_N L$, $C_N = 11 N (\ge 2A_N + 3)$,  then
$\rho(\Pi \bball_{L}(\Bx), \Pi \bball_{L}(\By)) \ge 2L$.
\end{enumerate}

\end{lemma}

\begin{proof}
(B) Consider the set
$\Pi \ball_{2L}(\Bx) \equiv \cup_{j=1}^n \ball_{2L}(x_j)$.
If $\diam  \Pi\Bx > A_N L$,
this union cannot be connected, for otherwise it would have diameter
$\le N\cdot  4 L$, but
$\diam \Pi \ball_{2L}(\Bx) \ge \diam \Pi\Bx > A_N L = 4 N L$.
Therefore, $\Pi \ball_{L^{1+\delta}}(\Bx)$ can be decomposed into a disjoint union
of two non-empty clusters,
$$
\Big( \cup_{j\in\cJ} \, \ball_{2L}(x_j) \Big)
\coprod
\Big( \cup_{i\in\cJ^c} \, \ball_{ 2L}(x_i) \Big),
$$
so for all $j\in\cJ$, $i\in\cJ^c$ one has
$$
\rho( \ball_L(x_i), \ball_L(x_j)) > 2L,
$$
yielding the assertion (B).

\par\smallskip
\noindent
(C) If $\rho(\Bx,\By) > C_N L$, then for some $\jknot$,
$\rho(x_\jknot, y_\jknot) > C_N L$. Since both balls are FI,
for all $j\in[1,n]$
$$
\rho(x_j, x_\jknot) \le \diam\, \Pi\Bx \le A_N L, \quad
\rho(y_j, y_\jknot) \le \diam\, \Pi\By < A_N L,
$$
so by triangle inequality, for all $i,j\in[1,n]$
$$
\bal
\rho(x_i, y_j) &\ge
\rho(x_\jknot, y_\jknot) - \rho(x_i, x_\jknot) - \rho(y_j, y_\jknot)
%
\ge C_N L - 2 A_N L \ge 3 L,
\eal
$$
yielding
$$
\rho(\Pi \boxx_{L}(\Bx), \Pi \boxx_{L}(\By)) >  3L - 2L = L.
$$
\end{proof}

Note that in the case where a decomposition of the form \eqref{eq:Ham.decomposable} is possible, it may be not unique. For example, with $N=3$ and an interaction of range $r_0=1$,
a particle configuration $\Bx=\{0, 10, 20\}$ can be decomposed into two non-interacting subsystems in three different ways. For our purposes, it will suffice to assume that one decomposition is canonically associated with every partially interactive ball.

Now we state a result on localization in PI balls for Hamiltonians with a short-range interaction. The reader may want to compare the elementary proof of Lemma \ref{lem:instead.PITRONS} below with that of Lemma 3 in \cite{CS09b} which is more involved.

\begin{lemma}\label{lem:instead.PITRONS}
Assume that the interaction $\BU$ has finite range $r_0\in[0,+\infty)$.
Consider an PI $N$-particle ball with canonical decomposition
$\bball^{(N)}_L(\Bu) = \bball^{(n')}_L(\Bu_\cJ) \times \bball^{(n'')}_L(\Bu_{\cJ^c})$. If the balls
$\bball^{(n')}_{L_{k-1}}(\Bu')$ and $\bball^{(n'')}_{L_{k-1}}(\Bu'')$ are \loc,
then $\bball^{(N)}_{L_k}(\Bu)$ is also \loc.
\end{lemma}

\proof
Since the interaction between subsystems labeled by $\cJ$ and by its complement
$\cJ^c$ vanishes, the operator $\BH$ admits the decomposition
\be\label{eq:lem.instead.PITRONS}
\BH_{\bball_{L_k}(\Bu)} = \BH_{\bball^{(n')}_{L_k}(\Bu_\cJ)} \otimes \one^{(n'')}
+ \one^{(n')} \otimes \BH_{\bball^{(n'')}_{L_k}(\Bu_{\cJ^c})},
\ee
so that the eigenfunctions $\BPsi_j$ of $\BH_{\bball_{L_k}(\Bu)}$ can be chosen in the form
$
\BPsi_{j} = \Bphi_{j'} \otimes \Bpsi_{j''},
$
where $\Bphi_{j'}$ are eigenfunctions  of $ \BH_{\bball^{(n')}_{L_k}(\Bu_\cJ)}$ and, respectively,
$\Bpsi_{j''}$ are eigenfunctions of $ \BH_{\bball^{(n'')}_{L_k}(\Bu_{\cJ^c})}$. Now the required decay bounds for the spectral projections
\be\label{eq:PI.tensor}
| \, \BPsi_{j}(\Bx) \BPsi_{j}(\By) \, |
= | \, \Bphi_{j'}(\Bx_{\cJ}) \Bphi_{j''}(\Bx_{\cJ^c}) \, \Bpsi_{j'}(\By_{\cJ})\Bpsi_{j''}(\By_{\cJ^c}) \, |
\ee
follow directly from the respective bounds on $\Bphi$ and $\Bpsi$, which are guaranteed by the \loc\, hypotheses on the projection balls  $\bball^{(n')}_{L_k}(\Bu_\cJ)$ and $\bball^{(n'')}_{L_k}(\Bu_{\cJ^c})$.
\qedhere

\smallskip

Given numbers $p,b>0$, introduce the following sequence:
\be\label{eq:def.p.n}
P(n,k) = P(n,k,p) = 2^{N-n} p(1+b)^k, \; n=1, \ldots, N; k\ge 0.
\ee

\begin{lemma}\label{lem:PI}
Suppose that for all $n\in[1, N-1]$ and for a given $k\ge 0$ the following bound holds for any
$n$-particle ball $\bball^{(n)}_{L_k}(\Bu)$:
$$
\pr{ \bball^{(n)}_{L_k}(\Bu) \text{ is } \mathnloc } \le  L_k^{-P(n,k)}.
$$
Then for any $N$-particle  PI ball $\bball^{(N)}_{L_k}(\Bu)$ one has
$$
\pr{ \bball^{(N)}_{L_k}(\Bu) \text{ is } \mathnloc } \le  2 L_k^{-2P(N,k) }
\le \frac{1}{4}  L_k^{-P(N,k)}.
$$
\end{lemma}

\proof
It follows from Lemma \ref{lem:instead.PITRONS} that a ball $ \bball^{(N)}_{L_k}(\Bu)$
admitting a decomposition of the form \eqref{eq:lem.instead.PITRONS} is \nloc\, only if at least one of the projection balls  $\bball^{(n')}_{L_k}(\Bu_\cJ)$, $\bball^{(n'')}_{L_k}(\Bu_{\cJ^c})$ is \nloc. Therefore, with $L_0$ large enough, we obtain
$$
\ba
\pr{  \bball^{(N)}_{L_k}(\Bu) \text{ is } \mathnloc } \\
 \quad
  \le \pr{\bball^{(n'')}_{L_k}(\Bu_{\cJ^c}) \text{ is } \mathnloc}
+ \pr{\bball^{(n'')}_{L_k}(\Bu_{\cJ^c})  \text{ is } \mathnloc}
\\
\quad  \le  L_k^{-P(n',k)} +  L_k^{-P(n'',k)}
\le 2 L_k^{-2P(N,k) }
 \quad    < \quart L_k^{-P(N,k)}.
\hfill\qed
\ea
$$

\subsection{Fully interactive $N$-particle balls }

Owing to Lemma \ref{lem:PI}, it suffices now to establish localization bounds for $N$-particle FI balls  $\bball_{L_{k+1}}(\Bx)$, assuming, if necessary, similar bounds
\begin{itemize}
  \item for $n$-particle balls  of any radius $L_{k'}$, $k'\ge 0$, with $n<N$,
  \item for $N$-particle balls  of any radius $L_{k'}$ with $k'\le k$.
\end{itemize}

\begin{lemma}[Main inductive lemma]\label{lem:DS.mixed}
Suppose that for a given $k\ge 0$ and all $i\in[0,k]$ the following bound holds for any
$N$-particle ball $\bball^{(N)}_{L_k}(\Bu)$:
\be\label{eq:lem.DS.mixed.hyp}
\pr{ \bball^{(N)}_{L_i}(\Bu) \text{ is } \mathnloc }
\le  L_i^{-P(N,i,p)}.
\ee
Then for any $N$-particle  FI ball $\bball^{(n)}_{L_{k+1}}(\Bx)$ one  has
\be\label{eq:lem.DS.mixed.claim}
\pr{ \bball^{(N)}_{L_{k+1}}(\Bu) \text{ is } \mathnloc } \le  L_{k+1}^{-P(N,k+1,p)},
\ee
provided that
\be\label{eq:cond.p.Np.mix}
\ba
p > \frac{2\alpha^2}{2-\alpha^2}Nd, \quad
0 < 3b \le
\min\left\{\frac{2 - \alpha^2}{\alpha^2} - \frac{2Nd}{p},
\sqrt{2}-1 \right\}.
\ea
\ee
Furthermore, for any pair of distant balls $\bball^{(N)}_{L_{k}}(\Bu)$,
$\bball^{(N)}_{L_{k}}(\Bv)$, one has
\be\label{eq:lem.DS.mixed.claim.2S}
\pr{\exists\,E\in \DR:\;
\text{ $\bball^{(N)}_{L_{k}}(\Bu)$,
$\bball^{(N)}_{L_{k}}(\Bv)$ are \EmS } }
\le  L_{k}^{-P(N,k,p)}.
\ee
\end{lemma}

\proof
$\,$
Set
$$
\bal
\cN_{k+1} &= \{ \bball^{(N)}_{L_{k+1}}(\Bu) \text{ is } \mathnloc \}, \\
\cS^{(2)}_{k} &= \{\exists\, E\, \exists\, \text{ distant \EmS\, balls } \bball^{(N)}_{L_{k}}(\Bu), \bball^{(N)}_{L_{k}}(\Bv) \subset \bball^{(N)}_{L_{k+1}}(\uw) \},
\\
\cR^{(2)}_{k} &= \{\exists\, E\, \exists\, \text{ distant \PR\, balls } \bball^{(N)}_{L_{k}}(\Bu), \bball^{(N)}_{L_{k}}(\Bv) \subset \bball^{(N)}_{L_{k+1}}(\uw) \}.
\\
\eal
$$
By virtue of Lemma \ref{lem:NT.implies.nloc.IID}, we have
$
\cN_{k+1} \subset
\cR^{(2)}_k \cup \Big( \cS^{(2)}_k \setminus \cR^{(2)}_k\Big),
$
and by Theorem \ref{thm:PCT}, $\pr{\cR^{(2)}_k} \le e^{-L_k^{\beta'}}$, so that it remains to assess the probability $\pr{ \cS^{(2)}_k \setminus \cR^{(2)}_k}$. Fix points $\Bu, \Bv\in\bball^{(N)}_{L_{k+1}}(\uw)$ and introduce the event
(figuring in assertion \eqref{eq:lem.DS.mixed.claim.2S})
$$
\cS^{(2)}_k(\Bu,\Bv) = \{ \exists\, E: \text{ $\bball^{(N)}_{L_{k}}(\Bu)$ and
 $\bball^{(N)}_{L_{k}}(\Bv)$ are \EmS\, } \}.
$$
Within the event $\cS^{(2)}_k(\Bu,\Bv) \setminus \cR^{(2)}_k$, either
$\bball^{(N)}_{L_{k}}(\Bu)$ or $\bball^{(N)}_{L_{k}}(\Bv)$ must be \CNR\,; w.l.o.g., assume that $\bball^{(N)}_{L_{k}}(\Bu)$ is  \CNR. Since it is
\EmS,  Lemma \ref{lem:CNR.and.no.2.S.imply.NS} implies that it must contain a pair of $\mathnloc$ balls $\bball^{(N)}_{L_{k-1}}(\By')$, $\bball^{(N)}_{L_{k-1}}(\By'')$ with $\rd(\By', \By'')> \half L_{k-1}^{1+\delta}$.
Consider two possible situations.
\par\smallskip\noindent
\textbf{(1)} {\sl Both $\bball^{(N)}_{L_{k-1}}(\By')$ and $\bball^{(N)}_{L_{k-1}}(\By'')$ are
\emph{FI}}.

The supports $\Pi \bball_{L_{k-1}}(\By')$, $\Pi \bball_{L_{k-1}}(\By'')$ of these
of fully interactive balls  have diameter smaller than $\quart L_k^{1+\delta}$ (cf. Lemma \ref{lem:FI.diam}, assertion (A)), so for $L_0$ large enough,
$$
\ba
\dist(\Pi \bball_{L_{k}}(\Bu), \Pi \bball_{L_{k}}(\Bv) ) > L_{k}^{1+\delta} - 2 \cdot \quart L_k^{1+\delta} \ge \half L_{k}^{1+\delta}.
\ea
$$
Using the inductive assumption \eqref{eq:lem.DS.mixed.hyp} and the mixing condition \Wtwo, we can write, with $L_{k-1} = L_{k+1}^{1/\alpha^2}$ and $L_0$ large enough:
$$
\pr{ \text{ $\bball^{(N)}_{L_{k-1}}(\By')$ and $\bball^{(N)}_{L_{k-1}}(\By'')$ are $\mathnloc$  }  }
\le 2 L_{k+1}^{-\frac{2p}{\alpha^2} (1+b)^{k-1}}
$$
The number of pairs $\By', \By''$ is bounded by
$|\bball^{(N)}_{L_{k}}(\uw)|^2/2 \le \frac{3^{2Nd} }{2}  L_{k}^{2Nd}$, so in this case,
\be\label{eq:prob.S2.not.R2}
\pr{ \cS^{(2)}_k \setminus \cR^{(2)}_k} \le \frac{3^{2Nd} }{2} \, L_{k+1}^{-\frac{2p}{\alpha^2}(1+b)^{k-1} +2Nd}.
\ee
A straightforward calculation shows that for $L_0$ large enough, the RHS is bounded by
$\quart L_{k+1}^{ -P(N,k+1,p)}$, under the conditions \eqref{eq:cond.p.Np.mix}. Indeed, we need the inequality
\be\label{eq:b.p.1}
\ba
\frac{2p}{\alpha^2} 2^{N-N}p(1+b)^{k-1}  - 2Nd > 2^{N-N} p (1+b)^{k+1}
\ea
\ee
or, equivalently,
$$
\ba
(1+b)^2 < \frac{2p}{\alpha^2}  - \frac{2Nd}{ p(1+b)^{k-1}}.
\ea
$$
Observe that for $b\in(0,1)$, $(1+b)^2 < 1 + 3b$, and
$\frac{2Nd}{p} > \frac{2Nd}{ p(1+b)^{k-1}}$. Therefore, if
$
p > \frac{2\alpha^2}{2-\alpha^2}Nd
$
and
$
3b \le \frac{2 - \alpha^2}{\alpha^2} - \frac{2Nd}{p},
$
then \eqref{eq:b.p.1} is also satisfied, so that
\be\label{eq:bound.S2}
\bal
\pr{ \cS^{(2)}_k } &\le
\pr{ \cR^{(2)}_k} + \pr{ \cS^{(2)}_k \setminus \cR^{(2)}_k}
\\
&\le e^{-L_k^{\beta'}} + {\textstyle \quart} L_{k+1}^{ -P(N,k+1,p)}
< {\textstyle \half} L_{k+1}^{ -P(N,k+1,p)}
\eal
\ee
yielding the assertion \eqref{eq:lem.DS.mixed.claim} in the case (1).

\smallskip
\par\noindent
\textbf{(2)}  {\sl Either $\bball^{(N)}_{L_{k-1}}(\By')$ or $\bball^{(N)}_{L_{k-1}}(\By'')$ is PI.} Then by Lemma \ref{lem:PI},
$$
\ba
\diy
\pr{ \text{ $\bball^{(N)}_{L_{k-1}}(\By')$ and $\bball^{(N)}_{L_{k-1}}(\By'')$ is $\mathnloc$  }  } \\
\diy
\quad \le \min \left(
      \pr{ \text{ $\bball^{(N)}_{L_{k-1}}(\By')$  is $\mathnloc$  }  },
      \pr{ \text{$\bball^{(N)}_{L_{k-1}}(\By'')$ is $\mathnloc$  }  }
\right)
\\
\diy \quad
\le  \,  L_{k-1}^{- 2p (1+b)^{k-1}}
<  L_{k+1}^{- p (1+b)^{k+1}},
\ea
$$
provided that $(1+b)^2 < 2$, i.e., $b < \sqrt{2} -1$. This completes the proof.
\qedhere
\smallskip

Applying Lemma \ref{lem:L.0.Np} and Lemma \ref{lem:PI}, we come by induction to the following result.

\begin{theorem}[Localization at any scale]\label{thm:loc.ind}
Assume that $|g|$ is large enough, so that \eqref{eq:L0.GF.Np}--\eqref{eq:L0.loc.Np} hold true
with $p > \frac{2\alpha^2}{2-\alpha^2}Nd$.
Then for any $k\ge 0$ and any $\Bu\in\DZ^{Nd}$
\be\label{eq:lem.DS.mixed.hyp.2}
\pr{ \bball^{(N)}_{L_k}(\Bu) \text{ is } \mathnloc }
\le {\textstyle \quart} L_k^{-P(N,k,p)} = {\textstyle \quart} L_k^{-p(1+b)^k}
\ee
with $b>0$ satisfying  \eqref{eq:cond.p.Np.mix}.
In other words, with probability $\ge 1 - \quart L_k^{-p(1+b)^k}$, for every  eigenfunction
$\Psi_j$ of operator $\BH_{\bball^{(N)}_{L_k}(\Bu)}$ and for all
$\Bx, \By \in \bball^{(N)}_{L_k}(\Bu)$ such that
$\|\Bx - \By \|\ge L_k^{\frac{1+\delta+\rho}{\alpha}}$ one has
$$
\left| \Psi_j(\Bx) \Psi_j(\By) \right| \le e^{-\gamma(m,L_k)L_k} < e^{-m L_k}.
$$
\end{theorem}

This marks the end of the direct scaling analysis of localization of eigenfunctions
in arbitrarily large finite balls.
In Section \ref{sec:SimpleDL.finite.L}, we will derive from the results of the scaling analysis uniform bounds on EF correlators in arbitrarily large finite balls.

\section{Scaling analysis: Adaptation to infinite-range interactions}
\label{sec:MPMSA.infinite.range}

In this section, we adapt a certain number of definitions and statements to rapidly decaying, infinite-range interactions satisfying the assumption \Uone.

\begin{enumerate}[\rm(S1)]
  \item Given a number $\delta\in(0,1/14)$, figuring in \Uone, set
$$
\varrho = 2\delta, \; \alpha = 1 + 4\delta, \; \tau = \delta/2,
$$
and observe that $\alpha^2<2$, $\varrho - \delta > \tau$.

  \item replace the decay exponent $\gamma(m,L)=m(1+L^{-\tau})$,
figuring in the definition \ref{def:S} (cf. Eqn \eqref{eq:def.gamma}),
by
$$
\gamma(m,L,n) := m(1+L^{-\tau})^{N-n+1}.
$$
  \item Modify the definition of distant balls as follows: balls
  $\bball_L(\Bx)$, $\bball_L(\By)$ are called distant iff
  $\rho(\Bx, \By)>C_N L^{1+\delta}$, $C_N = 4N$.
  \item Modify the definition of PI and FI balls: a ball $\bball_L(\Bx)$
  is called PI if $\diam(\Bx) = \diam(x_1, \ldots, x_N) > L^{1+\delta}$,
  and FI, otherwise.
  \item Check that if a ball $\bball^{(N)}_L(\Bx)$ is PI, then it admits
  a decomposition
$$
\bball^{(N)}_L(\Bx) = \bball^{(n')}_L(\Bx') \times \bball^{(n'')}_L(\Bx'')
$$
with $\dist( \Bx', \Bx'') > 2L^{1+\delta}$.

  \item For every positive integer $R$, introduce the following finite-range approximations $\BU_R$, $R\ge r_0$, of a given interaction $\BU$
generated by $2$-body interaction potentials $U^{(2)}(r)$:
$$
U^{(2)}_R(r)
= \one_{r\le R} \,U^{(2)}(r).
$$

  \item In order to measure the interaction between subconfigurations
  $\Bx'$, $\Bx''$ of an $N$-particle configuration $\Bx\in\cZp$, introduce the quantity
$$
\epsilon(R) := \sup_{\rd(\Bx',\Bx'')>R}
| \BU(\Bx) - \BU(\Bx') - \BU(\Bx'') |, \, \Bx=(\Bx',\Bx''),
$$
and check that for $R_k = C_N L_k^{1+\delta}$, one has
$$
\eps(R_k) < e^{-2mL_k} < \half e^{-L_k^{\beta}}.
$$
\end{enumerate}

The most important technical modification is required for the Lemma \ref{lem:instead.PITRONS},
the proof of which was very elementary, due to the use of eigenfunctions.

Next, introduce the following finite-range approximations $\BU_R$, $R\ge r_0$, of a given interaction $\BU$
generated by $2$-body interaction potentials $U^{(2)}(r)$:
$$
U^{(2)}_R(r)
= \one_{r\le R} \,U^{(2)}(r).
$$
Then it is readily seen that for any configuration $\Bx = (\Bx', \Bx'')$ with
$\dist(\Bx', \Bx'')> R$, one has $\BU_R(\Bx) = \BU_R(\Bx') + \BU_R(\Bx'')$. Furthermore, for any lattice subset of the form $\bball=\bball_{L'}(\Bx')\times\bball_{L''}(\Bx'')$
with $\dist(\bball_{L'}(\Bx'), \bball_{L''}(\Bx''))> R$, the Hamiltonian
$\BH^{(R)}_{\bball}$ admits the algebraic decomposition
$$
\BH^{(R)}_{\bball} = \BH^{(R)}_{\bball_{L'}(\Bx')} \otimes \one^{(n'')}
+ \one^{(n')} \otimes \BH^{(R)}_{\bball_{L''}(\Bx'')},
\; \Bx'\in\DZ^{dn'}, \Bx''\in\DZ^{dn''}.
$$

Next, we fix an arbitrary $m\ge 1$ and
modify the definition of a PI ball: $\bball^{(N)}_{L}(\Bx)$ is called PI iff
$\diam(\Bx) > C_N = \frac{4mN}{\tm}$,
and, therefore, it admits a decomposition
$\bball^{(N)}_{L}(\Bx)=\bball_{L'}(\Bx')\times\bball_{L''}(\Bx'')$
with $\dist( \Bx', \Bx'') > \frac{C_N}{2N}L = \frac{2mL}{\tm}$,
so that, for $R_k = \frac{C_N}{2N}L_k$, one has
$$
\eps(R_k) < e^{-2mL_k}  < \half e^{-L_k^{\beta}}.
$$

\begin{lemma}\label{lem:instead.PITRONS.decaying.33}
Assume that the interaction $\BU$ satisfies the condition \Uone. Fix an energy $E$.
Let an $N$-particle ball
$\bball^{(N)}_{L_k}(\Bu)
= \bball^{(n')}_{L_k}(\Bu_\cJ) \times \bball^{(n'')}_{L_k}(\Bu_{\cJ^c})$,
and a sample of the random potential $V(\cdot;\omega)$ be such that
\begin{enumerate}[{\rm(a)}]
  \item $\dist( \Bx', \Bx'') > R_k := \frac{C_N}{2N}L_k$;

%
  \item $\bball^{(n')}_{L_{k}}(\Bu')$ and $\bball^{(n'')}_{L_{k}}(\Bu'')$ are
  \loc;
  \item $\bball^{(N)}_{L_k}(\Bu)$ is \CNR.
\end{enumerate}
Then the ball $\bball^{(N)}_{L_k}(\Bu)$ is \EmNS.
\end{lemma}

\proof
Set $R_k = \frac{C_N}{2N}L_k$ and define $\BU^{(R_k)}$ and $\BH^{(R_k)}$ as above.
Denote
$$
\BH^{(R_k,N)}_{\B{u},k}=\BH^{(R_k)}_{\bball_{L_{k}}^{(N)}(\B{u})},\;
\BG(E) = (\BH_{\bball_{L_{k}}^{(N)}(\B{u})} - E)^{-1},\;
\BG^{(R_k)}(E) = (\BH^{(R_k,N)}_{\B{u},k} - E)^{-1}.
$$
Operator $\BH^{(R_k,N)}_{\B{u},k}$ admits the decomposition
$$
\BH^{(R_k,N)}_{\B{u},k} =
\BH_{\bball^{(n')}_{L_{k}}(\Bu')} \otimes \one^{(n'')}
+ \one^{(n')} \otimes \,\BH_{\bball^{(n')}_{L_{k}}(\Bx'')},
\; \Bu=(\Bu',\Bu'')\in\DZ^{dn'}\times\DZ^{dn''}
$$
thus its eigenvalues are the sums $E_{a,b}=\lam_a+\mu_b$, where
$\{\lam_a\}=\Sigma(\BH_{\bball_{L_{k}}(\Bu')})$ is the spectrum of
$\BH_{\bball_{L_{k}}(\Bu')}$ and, respectively,  $\{\mu_b\}=\Sigma(\BH_{\bball_{L_{k}}(\Bu'')})$.
Eigenvectors of $\BH^{(R_k,N)}_{\B{u},k}$ can be chosen in the form
$
\BPsi_{a,b} = \Bphi_a \otimes \Bpsi_b
$
where $\{\Bphi_a\}$ are eigenvectors of $\BH_{\bball^{(n')}_{L_{k}}(\Bu')}$ and
$\{\Bpsi_b\}$ are eigenvectors of $\BH_{\bball^{(n'')}_{L_{k}}(\Bu'')}$.
Note that for each eigenvalue $E_{a,b}=\lam_a+\mu_b$, i.e., for each pair $(\lam_a,\mu_b)$, the non-resonance assumption $|E - (\lam_a + \mu_b)|\ge e^{L_k^\beta}$
reads as $|(E -\lam_a) - \mu_b)|\ge e^{L_k^\beta}$ and also as
$|(E -\mu_b) - \lam_a)|\ge e^{L_k^\beta}$. Therefore, we can write
\be
\bal
\BG^{(R_k)}(\Bu,\By; E) &= \sum_{\lam_a} \sum_{\mu_b}
\frac{ \Bphi_a(\Bu') \Bphi_a(\By')\, \Bpsi_b(\Bu'') \Bpsi_b(\By'')\, }
{ (\lam_a + \mu_b) - E}
\\
& = \sum_{\lam_a} \BP'_a(\Bu',\By') \, \BG^{(R_k)}_{\bball_{L_k}(\Bu'')}(\Bu'',\By''; E-\lam_a)
\\
& = \sum_{\mu_b}  \BP''_b(\Bu'',\By'') \, \BG^{(R_k)}_{\bball_{L_k}(\Bu')}(\Bu',\By'; E-\mu_b),
\eal
\ee
where the resolvents $\BG^{(R_k)}_{\bball_{L_k}(\Bu')}(E-\mu_b)$ and
$\BG^{(R_k)}_{\bball_{L_k}(\Bu'')}(E-\lam_a)$ are non-resonant:
$$
\| \BG^{(R_k)}_{\bball_{L_k}(\Bu')}(E-\mu_b) \| \le e^{L_k^{\beta}},
\;\; \| \BG^{(R_k)}_{\bball_{L_k}(\Bu'')}(E-\lam_a) \| \le e^{L_k^{\beta}}.
$$
For any $\By\in\pt^- \bball_{L_k}(\Bu)$, either
$\rho(\Bu',\By') = L_k$, in which case we infer from  (b)
\be
\big| \BG^{(R_k)}(\Bu,\By; E) \big| \le |\bball_{L_k}(\Bu'')|\, e^{-\gamma(m,L_k,N-1)L_k + L_k^{\beta}}
\ee
or $\rho(\Bu'',\By'') = L_k$, and then we have, respectively,
$$
\big| \BG^{(R_k)}(\Bu,\By; E) \big| \le |\bball_{L_k}(\Bu')|\, e^{-\gamma(m,L_k,N-1)L_k + L_k^{\beta}}.
$$
In either case, the LHS is bounded by
$$
\ba
\exp\left( -m(1+L_k^{-\tau})^{N-(N-1)+1} L_k + L_k^{\beta} + \Const \ln L_k \right)
< \half e^{ -\gamma(m,L_k,N)}.
\ea
$$
Next, by the second resolvent identity,
$
\BG = \BG^{(R_k)} - \BG^{(R_k)} (\BU - \BU^{(R_k)}) \BG,
$
and using the assumed \NR\, property of the resolvents $\BG$, $\BG^{(R_k)}$, we conclude that
$$
\ba
\|\BG - \BG^{(R_k)} \|  \le \|\BU - \BU^{(R_k)}\| \, \|\BG^{(R_k)}\| \, \| \BG \|
 \le e^{-2mL_k}  \, e^{2L_k^{\beta}} \le \half e^{-\gamma(m,L_k,N)} .
\ea
$$
Now the claim follows from the inequality
$$
\ba
\forall\, \Bx,\By\in \bball_{L_{k+1}}^{(N)}(\B{u}) \quad
\big| \BG (\Bx, \By;E) - \BG^{(R_k)} (\Bx, \By;E) \big| \le \half e^{ -\gamma(m,L_k,N)}.
 \quad{}_{\qed}
\ea
$$

\begin{lemma}\label{lem:two.PI.decaying}
Let $\bball_{L_{k+1}}(\Bx), \bball_{L_{k+1}}(\By)$ be two
distant {\rm PI} balls. Then, for $L_0$ large enough,
\be
\ba
\pr{ \text{ $\bball_{L_{k+1}}(\Bx)$ and $\bball_{L_{k+1}}(\By)$ are \EmS\,  } }
< \frac{1}{4} L_{k+1}^{-P(N,k+1)}.
\ea
\ee
\end{lemma}
\proof
Let
$$
\bal
\cR^{(2)}
  &= \{ \exists\, E:\,  \text{ $\bball_{L_{k+1}}(\Bx)$ and $\bball_{L_{k+1}}(\By)$ are \ER\,  } \}
\\
\cS^{(2)}
  &= \{   \text{ $\bball_{L_{k+1}}(\Bx)$ and $\bball_{L_{k+1}}(\By)$ are \EmS\, } \}
\\
\eal
$$
Then $\pr{ \cS^{(2)} } \le \pr{ \cR^{(2)} } + \pr{ \cS^{(2)}\setminus\cR^{(2)} }$, and within
the event $\cS^{(2)}\setminus\cR^{(2)}$ one of the balls  $\bball_{L_{k+1}}(\Bx)$,
$\bball_{L_{k+1}}(\By)$ must be \NR, no matter how $E\in\DR$ is chosen.

Next, consider the canonical decomposition of the PI ball
$$
\bball_{L_{k+1}}(\Bx)=\bball_{L_{k+1}}(\Bx')\times\bball_{L_{k+1}}(\Bx'')
$$
By Lemma
\ref{lem:instead.PITRONS.decaying}, if both $\bball_{L_{k+1}}(\Bx')$ and
$\bball_{L_{k+1}}(\Bx'')$ are \loc\, and $\bball_{L_{k+1}}(\Bx)$ is \EmS, then
$\bball_{L_{k+1}}(\Bx)$ is \ER, thus $\bball_{L_{k+1}}(\By)$ is \NR. On the other hand,
consider  the canonical decomposition of the PI ball
$$
\bball_{L_{k+1}}(\By)=\bball_{L_{k+1}}(\By')\times\bball_{L_{k+1}}(\By'').
$$
If both $\bball_{L_{k+1}}(\By')$ and
$\bball_{L_{k+1}}(\By'')$ are \loc\, and $\bball_{L_{k+1}}(\By)$ is \EmS, then
$\bball_{L_{k+1}}(\By)$ must be \ER. Let
$$
\cL^{(4)} =
\{
\text{ one of the balls  } \bball_{L_{k+1}}(\Bx'), \bball_{L_{k+1}}(\Bx''),
\bball_{L_{k+1}}(\By'), \bball_{L_{k+1}}(\By'') \text{ is \nloc }
\}.
$$
Since these four balls  correspond to systems with  $\le N-1$ particles, we can use induction in $N$ and write
$$
\pr{ \cL^{(4)} } \le 4 L_{k+1}^{-P(N-1,k+1)} = 4 L_{k+1}^{-2P(N,k+1)}.
$$
Then, as we have noticed, $\cS^{(2)}\setminus\cL^{(4)} \subset\cR^{(2)}$, hence,
for $L_0$ large enough,
$$
\ba
\pr{ \cS^{(2)} } \le \pr{ \cR^{(2)} } + \pr{ \cL^{(4)} }
\le  e^{-L_k^{\beta}} + 4 L_{k+1}^{-2P(N,k+1)}
< \frac{1}{4} L_{k+1}^{-P(N,k+1)}.
\ea
$$
\qedhere

The statement of Lemma \ref{lem:CNR.and.no.2.S.imply.NS} remains unchanged, but its proof
requires a minor modification.
\begin{proof}(\emph{Lemma \ref{lem:CNR.and.no.2.S.imply.NS} under the hypothesis \Uone.})
\noindent
(A) By assumption, either $ \bball_{L_{k+1}}(\Bu)$ is \CNR\, and contains no \EmS\, ball
of radius $L_{k}$, or there is a point $\Bw \in\bball_{L_{k+1}}(\Bu)$ such that
any ball $\bball_{L_{k}}(\Bv)\subset\bball_{L_{k+1}}(\Bu)$ with
$\rho(\Bw,\Bv)\ge C_N L^{1+\delta}_{k}$
is \EmNS. In the former case, such an exclusion is unnecessary, but in order to treat
both situations with one argument, we can formally set $\Bw = \Bu$ (or any other point).

Fix points $\Bx,\By$ with
$R:=\rho(\Bx,\By)>L_{k}^{1+\varrho}$.
By triangle inequality,
$$
\textstyle
\rho(\Bx, \bball_{(C_N - 1) L_{k}}(\Bw))  +
\rho(\By, \bball_{ (C_N - 1) L_{k}}(\Bw))
\ge R - (2C_N - 2) L^{1+\delta}_{k}.
$$
Assume first that
$$
\bal
r'  &:= \rho(\Bx, \bball_{ (C_N-1) L^{1+\delta}_{k}}(\Bw)) \ge  L_{k} +1,
\\
r'' &:= \rho(\By, \bball_{(C_N - 1) L^{1+\delta}_{k}}(\Bw)) \ge  L_{k} +1.
\eal
$$
All balls  of radius $L_{k}$ both in $\bball_{r'}(\Bx)$ and in $\bball_{r''}(\By)$ are
automatically \EmNS.
Furthermore,
$r' + r'' \ge R - 2(C_N-1)L^{1+\delta}_{k} - 2
\ge R - 2C_N L^{1+\delta}_{k}$.

Consider the set
$\bball =  \bball_{r'}(\Bx)  \times \bball_{r''}(\By)$ and the function  $f: \bball \to \DC$ defined by
$f(\Bx', \Bx'') = \BG_{\bball_{L_{k+1}}}(\Bx', \Bx''; E)$.
Since $E$ is not a pole of the resolvent $\BG_{\bball_{L_{k+1}}(\Bu)}(\cdot)$, it is well-defined
(hence, bounded, on a finite set).
By Lemma \ref{lem:cond.SubH.IID}, $f$ is  $(L_{k}, q)$-subharmonic both in $\Bx'$ and in $\Bx''$, with
$q \le e^{-\gamma(m, L_{k},n)}$.
Therefore, one can write, with the convention $-\ln 0 = +\infty$, using Lemma \ref{lem:BiSubH.IID} and setting for brevity $J := M - n+1$:
\be\label{eq:proof.lem.CNR.and.no.2.S.imply.NS}
\bal
-\ln f(\Bu, \By)
&\ge
-\ln \Bigg[ \left(e^{ -m(1+ \half L_{k}^{-\tau})^J L_{k} }
\right)^{\frac{R - 2C_N L^{1+\delta}_{k}- 2 L_{k}  }{L_{k}+1}}
e^{L_{k+1}^\beta}
\Bigg]
\\
& = m\big(1+ {\textstyle \half } L_{k}^{-\tau} \big)^J
    {\textstyle \frac{L_{k}}{ L_{k} + 1} } R\left(1 - 3C_N R^{-1} L_{k} \right)
    - L_{k+1}^\beta
\\
&
\textstyle
 = mR \left[ \left(1+ {\textstyle \half } L_{k}^{-\tau} \right)^J
   (1 -  L^{-1}_{k})  \left(1 - 3C_N L_{k}^{-\varrho } \right)
    - \frac{L_{k+1}^\beta}{mR} \right]
\\
&\ge m\left(1+ {\textstyle \quart} L_{k}^{-\tau} \right)^J R
\ge \gamma(m,L_{k+1}) R.
\eal
\ee
If $r'=0$ (resp., $r''=0$), the required bound follows from the subharmonicity of the function $f(\Bx',\Bx'')$ in $\Bx''$ (resp., in $\Bx'$).

%
\par\smallskip\noindent
(B) This assertion is proved in the same way as in Section \ref{ssec:tun.and.loc.finite.balls}.
\end{proof}

\vskip1mm
\emph{ The rest of the scaling procedure presented in Section \ref{sec:MPMSA.finite.range} does not have to be modified and applies to infinite-range interactions.}

\section{Strong dynamical localization in finite volumes}
\label{sec:SimpleDL.finite.L}

\subsection{Uniform bounds in finite volumes}

We use here a finite-volume variant of a `soft' argument proposed by Germinet and Klein
in \cite{GK01}; working with finite cubes allows to avoid a functional-analytic complement
concerning the weighted Hilbert-Schmidt norms of spectral projections of operators
$\BH(\om)$ in the entire graph $\cZp$ and replace it with a simple application of Bessel's inequality.

Denote by $\csB_1(I)$ the set of all Borel functions $\phi:\DR\to\DC$ with
$\supp\,\phi\subset I$ and $\|\phi\|_\infty \le 1$.

\begin{theorem}\label{thm:GK}
 Fix an integer $L\in\DN^*$ and assume that  the following bound holds for any pair of disjoint balls $\bball_L(x), \bball_L(y)$:
$$
\pr{ \exists\, E\in I:\, \text{ $\bball_L(\Bx)$ and $\bball_L(\By)$ are $(E,m)$-S} } \le f(L).
$$
Then for any $\Bx,\By\in\cZp$ with $\rd(\Bx,\By)> 2L+1$, any connected subset
$\BLam\supset\bball_L(x) \cup \bball_L(y)$ and any Borel function $\phi\in\csB_1(I)$
\be\label{eq:thm.MSA.to.DL}
\esm{ \langle \big|\one_{\Bx} | \phi(\BH_\BLam(\om)) | \one_{\By} \rangle \big| }
\le CL^d \eul^{-mL} + f(L).
\ee
\end{theorem}

\proof Fix points $\Bx,\By\in\cZp$ with $\rd(\Bx,\By)> 2L+1$ and a finite connected graph
$\BLam\supset\bball_L(\Bx) \cup \bball_L(\By)$. The operator $H_\BLam(\om)$
has a finite orthonormal eigenbasis $\{\psi_i\}$ with respective eigenvalues
$\{\lam_i\}$. Set $\BS = \pt \bball_L(\Bx) \cup \pt \bball_L(\By)$ (recall: this is a set of \emph{pairs}
$(u,u')$). Suppose that for some $\om$, for each $i$ there is $\Bz\in \{\Bx,\By\}$ such that
$\bball_L(\Bz_i)$ is $\lam_i,m)$-NS; let $\{\Bv_i \}= \{\Bx,\By\}\setminus \{\Bz_i\}$.
Denote
$\mu_{\Bx,\By}(\phi) = \langle \big|\one_{\Bx} | \phi(H_\BLam(\om)) | \one_{\By} \rangle \big|$,
with $\mu_{\Bx,\By}(\phi)\le 1$.
Then by the GRI for eigenfunctions,
and by Bessel's inequality used at the last stage of derivation,
$$
\bal
 \mu_{\Bx,\By}(\phi)
& \le \|\phi\|_\infty \, \sum_{\lam_i \in I} |\BPsi_i(\Bx) \BPsi_i(\By)|
\le \sum_{\lam_i \in I} |\BPsi_i(\Bz_i) \BPsi_i(\Bv_i)|
\\
& \le \sum_{\lam_i \in I} |\BPsi_i(\Bv_i)| \eul^{-mL}
      \sum_{(\Bu,\Bu')\in\pt \bball_L(\Bz_i)} |\BPsi_i(\Bu)|
\qquad\qquad\qquad\qquad\qquad\qquad
\\
& \le \eul^{-mL} \sum_{\lam_i \in I} \;\sum_{(\Bu,\Bu')\in \BS}
      |\BPsi_i(\Bu)| \left(|\BPsi_i(\Bx)| + |\BPsi_i(\By)|  \right)
%
\eal
$$
$$
\bal
&\le \eul^{-mL} |\BS| \, \max_{\Bu\in\cG} \sum_{\lam_i \in I}
    \half \left( |\BPsi_i(\Bu)|^2 + |\BPsi_i(\Bx)|^2 + |\BPsi_i(\By)|^2 \right)
\\
& \le  \eul^{-mL} \frac{|S|}{2} \, \max_{\Bu\in\BLam}
       \left( 2\|\one_\Bu\|^2 + \|\one_\Bx\|^2 + \|\one_\By\|^2 \right)
%
= \eul^{-mL} |\BS| \cdot 2
\eal
$$
where $|\BS|\le C L^d$.
Denote
$\cN_L = \myset{\exists\, E\in I:\, \text{ $\bball_L(\Bx)$ and $\bball_L(\By)$ are $(E,m)$-S}}$,
with $\pr{\cN_L}\le f(L)$, by assumption. Further,
$$
\esm{ \mu_{\Bx,\By}(\phi) } = \esm{ \one_{\cN_L}\mu_{\Bx,\By}(\phi) }
   + \esm{ \one_{\cN_L}\mu_{\Bx,\By}(\phi) }  \le f(L) + 2CL^d \eul^{-mL}.
\qedhere
$$

\qedhere
\vskip2mm

\subsection{Strong dynamical localization in the entire graph $\cZp$}
\label{ssec:SimpleDL.lattice}

Here we follow the same path as in earlier works by Aizenman et al. (cf., e.g.,
\cite{A94}, \cite{ASFH01}).

\begin{theorem}\label{thm:DL.lattice.IID}
Consider the Hamiltonian $\BH(\om)$ of the form \eqref{eq:H} with random potential satisfying the assumptions \Wone--\Wthree\, and the interaction potential satisfying one of the assumptions \Uzero, \Uone. Fix an interval $I\subset\DR$.  There is $g_0<+\infty$ such that
if $|g|\ge g_0$, then for all $x,y\in\DZ^d$, $x\ne y$,  and with the same $c, a>0$ as in \eqref{eq:c.a},
\be\label{eq:DL.bound.any.s.lattice.IID}
\esm{ \sup_{\|\phi\|_{\infty}\le 1} \,
\Big| \langle \one_{\Bx} \,|\, \phi(\BH_{}(\omega)) \, P_I(\BH_{}(\omega)) \,|\, \one_{\By}\rangle \Big| }
 \le \Const\, \eul^{-a \ln^{1+c} \rd(\Bx,\By) }.
\ee
\end{theorem}

\proof
For any  ball $\bball$ and any points $\Bx,\By\in\bball$ introduce a spectral measure $\bmubxy$ uniquely defined by
$$
\int\, \phi(\lambda)\, d\bmubxy(\lambda) = \bra{\one_\Bx} \phi(H_\bball(\omega))
\Pi_I(H_{\bball}(\omega)) \ket{\one_\By},
$$
where  $\phi$ is an arbitrary bounded Borel (or continuous) function,
and similar spectral measures $\bmuxy$ for the operator $\BH(\om)$ on the entire graph $\cZp$.
If $\{\bball_{L_k}\}$ is a growing sequence of balls, then $\bmukxy$ converge vaguely to $\muxy$ as $k\to\infty$. Note that this fact remains true in a much more general context of unbounded operators: by a well-known result (cf., e.g., \cite{Kato}),
for the strong resolvent convergence of operators $\BH_n \to\BH$, with a common core $\cD$,
it suffices that $\BH_n \phi \to \BH\phi$ for any element $\phi\in\cD$; in turn, this implies
the vague convergence of the spectral measures
$\langle {\boldsymbol{\varphi}}, \phi(\BH_n) {\boldsymbol{\psi}} \rangle$.
Usually, an appropriately chosen subspace of compactly supported functions can serve as a core,
and on such functions operators $\BH_n$ converge by stabilization.

So, by virtue of Fatou lemma on convergent measures, for any measurable set $\cE\subset\DR$
and any growing sequence of balls $\bball_{L_k}$,
$$
\esm{ |\muxy| (\cE) } \le \liminf_{k\to\infty} \; \esm{ |\mu^{\Bx,\By}_{\bball_{L_k}}| (\cE) }.
$$
Therefore, the uniform bounds  in finite volumes $\bball_{L_k}$, established in Theorem \ref{thm:MPMSA.implies.DL}, imply the dynamical localization on the entire lattice.
\qedhere
\smallskip

Taking functions $\phi_t: \lambda\mapsto e^{it\lambda}$, $t\in\DR$, one infers from
Theorem \ref{thm:DL.lattice.IID} a more traditional form of dynamical localization:

\begin{theorem}\label{thm:DL.lattice.eit}
Under the assumptions of Theorem \ref{thm:DL.lattice.IID}, there exist $a,c>0$ such that
for any finite subset $K\subset\DZ^d$ and any finite interval $I\subset\DR$
\begin{equation}\label{eq:thm.DL.lattice.IID.propagator}
\expect\left[ \sup_{t\in\D{R}} \;\left\| e^{a \ln^{1+c}\B{X}} \, \eul^{-\ii tH(\omega)}
          P_{I }(H(\omega))
\one_{K}\right\|\right] < \infty.
\end{equation}
\end{theorem}

Taking into account RAGE theorem(s)\footnote{See, e.g., the original papers \cite{AG73,E78} and their discussion in \cite{CFKS87}}, Theorem \ref{thm:DL.lattice.eit} implies the spectral localization: with probability one, the spectrum of $\BH(\om)$ is pure point, so that spectrally a.e. eigenfunction $\BPsi$ of $\BH(\om)$ is square-summable. However,
\eqref{eq:thm.DL.lattice.IID.propagator} does not imply directly an exponential decay of $\BPsi$.

\subsection{Exponential decay of eigenfunctions on the entire graph $\cZp$}
\label{ssec:SimpleSL.lattice}

\begin{theorem}\label{thm:Main.SL}
For $\DP$-a.e. $\om\in\Om$ every normalized eigenfunction $\BPsi$
of operator $\BH(\om)$ satisfies the following bound: for some $R(\om)$,
$\hBx(\om)$ and all
$\By$ with $\| \By \|\ge R(\om)$
\be\label{eq:thm.Main.SL}
  |\BPsi(\By)| \le e^{-m\| \By \|}.
\ee
\end{theorem}
\proof
Pick an arbitrary vertex $\Bz\in\cZp$.
By  Borel--Cantelli lemma, there is a subset $\Om'\subset\Om$ with $\pr{\Om'}=1$ such that for any $\om\in\Om'$ and some $k_0(\om)$,  all
$k\ge k_0$ and any $E\in\DR$ there is no pair of  $L_k^{1+\delta}$-distant \EmS\, balls  $\bball_{L_k}(\Bx), \bball_{L_k}(\By)\subset \bball_{L_{k+2}}(\Bz)$. Fix $\om\in\Om'$.

Let $\BPsi$ be a normalized eigenfunction of $\BH(\om)$ with eigenvalue $\lam$.
Since $\|\BPsi\|_2\le 1$ implies $\|\BPsi\|_\infty\le 1$, there is a point $\hBx$ such that
$\|\BPsi\|_2 = |\BPsi(\hBx)|$.
If $\hx_n\in\bball_{L_{k-1}}(0)$, then $\bball_{L_k}(0)$ must be $(\lam,m)-$S,
otherwise the $(\lam,m)$-NS property would lead to a contradiction:
$$
\|\BPsi\|_\infty = |\BPsi(\hBx)| \le e^{-mL_k} \|\BPsi\|_\infty < \|\BPsi\|_\infty.
$$
Thus any ball $\bball_{L_k}(\By)\subset \bball_{L_{k+2}}(\Bz)$ with
$\rho(\By,\Bz) \in [L_{k+1}, L_{k+2})$ (hence, with
$\rho(\By,\Bz) \ge L_k^\alpha > 2L_k^{1+\delta}$) is $(\lam,m)-$NS.
Note that the function $\Bx\mapsto |\BPsi_n(\Bx)|$ is $(L_k,q)$-subharmonic in $\bball_{R}(\By)$, with $R = \rho(\By,\Bz) - 2L_k^{1+\delta}-1 > \half L_k^{1+2\delta}$
and $q = e^{-\gamma(m,L_k)L_k}$. Recall that we set $\tau=\delta/4$; cf. \eqref{eq:relat.param.induct.fixed}.  Now Lemma \ref{lem:SubH.IID}
implies, for $L_k$ large enough,
$$
\ba
 -\frac{\ln |\BPsi(\By)|}{ \rho(\By,\Bz) }
\ge m \left(1+L_k^{-\tau} \right) \left(1 - \frac{2L_k+1}{ \rho(\By,\Bz) }\right)
\ge m \left(1 + \half L_k^{-\tau} \right) > m
\ea
$$
yielding the assertion \eqref{eq:thm.Main.SL}.
\qedhere

\section{Fermionic Hamiltonians on more general graphs}
\label{sec:general.graph}

The reduction to a standard lattice Laplacian on a subset
$\{(x_1, \ldots, x_N):\; x_1 < \cdots < x_N\}$ with Dirichlet boundary conditions
is no longer possible for particle systems on lattices $\DZ^d$ with $d>1$. Instead,
one has to work with a symmetric power of the lattice, considered as a graph.
So it seams reasonable to consider a fairly general, countable connected graph
$\cZ$ satisfying the condition of polynomial growth of balls:
$$
\forall\, x\in\cZ\;\; \forall\, L\ge 1\;\;
|\ball_L(x)| \le C_d L^d,
$$
where $\ball_L(x) = \{y\in\cZ:\, \rd_\cZ(x,y)\le L\}$. In particular, this gives a uniform
bound on the coordination numbers, $n_\cZ(x) \le C_d$ (of course, this bound
may be non-optimal).

An orthonormal basis in the Hilbert space of square-summable antisymmetric functions
$\BPsi:\cZ^N\to\DC$ is formed by the functions
$$
\BPhi_\Ba = \frac{1}{\sqrt{N!}} \sum_{\pi\in\fS_N} \myotimes_{j=1}^{N}
\one_{a_{\pi^{-1}(j)}},
\quad \Ba = \{a_1, \ldots, a_N\}, \;
\#\{a_1, \dots, a_N\}=N.
$$

Further, define the graph $(\cZp, \bcEp)$ as follows: the vertex set is
$$
\cZp
= \{
\Ba = \{a_1, \ldots, a_N\}: \; a_j\in\cZ, \#\{a_1, \dots, a_N\}=N
\}.
$$
Two vertices $\Ba, \Bb$ form an edge iff
\begin{itemize}
  \item the symmetric difference $\Ba \ominus \Bb$ has cardinality $1$,
  i.e. $\Ba = \{a_1, c_2, \ldots, c_N\}$, $\Bb = \{b_1, c_2, \ldots, c_N\}$,
  with $\#\{a_1, b_1, c_2, \ldots, c_N\}=N+1$, and
  \item $\rd_\cZ(a_1, b_1) = 1$.
\end{itemize}

Next, define on $\cZp$ the max-distance,
$$
\rho(\Bx, \By) =  \min_{\pi\in\fS_N} \max_{1\le j \le N} \rd_\cZ(x_{\pi^{-1}(j)}, y_j)
$$
and introduce the balls $\bball_L(\Bx)$ relative to the distance $\rho(\cdot\,,\cdot)$.

Now one can define the fermionic (negative) Laplacian $(-\BDelta)$ on $\cZp$
and random Hamiltonians
$\BH(\om) = -\BDelta + g\BV(\Bx;\om) + \BU(\Bx)$, where
$$
\BU(\Bx) = \sum_{i\ne j} U(\rd(x_i, x_j)),
\quad U:\DN\to\DR,
$$
and the external random potential energy
$$
\BV(\Bx;\om) = V(x_1;\om) + \cdots + V(x_N;\om),
$$
is generated by a random field $V:\cZ\times\Om\to\DR$.

The method presented in Sections \ref{sec:deterministic.bounds}--\ref{sec:MPMSA.infinite.range}
applies to strongly disordered random
Hamiltonians $\BH(\om) = -\BDelta + g\BV(\Bx;\om) + \BU(\Bx)$
describing fermionic systems on connected graphs $\cZ$ with polynomial
growth of balls; indeed, one can see that we did not use particular properties
of the one-dimensional lattice $\cZ=\DZ$.

Strongly disordered
bosonic systems can be treated in a similar way; the only modification required here
concerns the explicit form of the matrix elements of the Laplacian, which remains
a second-order finite-difference operator.


\section*{Appendix. Proofs of auxiliary statements}

\subsection{Proof of Lemma \ref{lem:L.0.Np} }

We start with the second assertion.

The random potential energy $\BV(\om)$ reads as follows:
\be\label{eq:BV.Nx}
\BV(x_1, \ldots, x_N;\om) = \sum_{y\in \Bx} \Bn_y V(y;\om).
\ee
Therefore, if $\rho(\Bx,\By)\ne 0$, then there exists a point
$w\in \Pi \bball_L(\Bu)$
such that $\Bn_w(\Bx)  \ne \Bn_w(\By)$. As a result,
\be\label{eq:decomp.V}
\BV(\Bx;\om) - \BV(\By;\om) = (\Bn_w\Bx) - \Bn_w(\By)) V(w;\om)
+ \sum_{v \ne w} c_v V(v;\om),
\ee
where the explicit form of the integer coefficients
$
c_v = \Bn_v(\Bx) - \Bn_v(\By)
$
is irrelevant for our argument:  it suffices to know that the sum in the RHS of Eqn.~\eqref{eq:decomp.V} is measurable with respect to the sigma-algebra $\fF_{\ne w}$ generated by the random variables $\{V(v;\cdot), v\ne w\}$, while  $\Bn_w(\Bx) - \Bn_w(\By) =: c_w \ne 0$,
$|c_w|\ge 1$. Therefore,
$$
\bal
\pr{ |g\BV(\Bx;\om) - g\BV(\By;\om)| \le s }
& = \esm{ \pr{ |\BV(\Bx;\om) - \BV(\By;\om)| \le |g|^{-1} s\,|\, \fF_{\ne w} } } \\
&= \esm{ \pr{ |c_w V(w;\om) + \zeta(\om)| \le |g|^{-1}\,|\, \fF_{\ne w} } } \\
\eal
$$
with some $\fF_{\ne w}$-measurable random variable $\zeta(\om)$
$$
\ba
= \esm{ \pr{ c_w V(w;\om) \in
[\zeta(\om)-|g|^{-1}s, \zeta(\om) + |g|^{-1}s] \,|\, \fF_{\ne w} } }. \\
\ea
$$
Since $|c_w|\ge 1$ and the interval $[\zeta(\om)-s, \zeta(\om) + s]$
 has length $2s$, we conclude that
$$
\ba
\displaystyle \pr{ |\BV(\Bx;\om) - \BV(\By;\om)| \le |g|^{-1}s } \le
\sup_{a\in \DR} \; ( F_{V,w}(a+2s) - F_{V,w}(a)).
\ea
$$
Since $F_{V,w}$ is continuous, by assumption \Wone, the latter quantity vanishes as
$|g|^{-1}s\downarrow 0$. Therefore,
$$
\pr{ \exists\, \Bx,\By\in \bball_L(\Bu):\, \Bx\ne \By, \; |g\BV(\Bx) - g\BV(\By)| \le s }
\tto{|g|\to \infty} 0,
$$
so, with arbitrarily high probability,
the spectrum of the diagonal operator $\BV(\om)$ in $\bball_{L_0}(\Bu)$ admits a
positive uniform lower bound $s>0$ on spectral spacings (differences between the eigenvalues).
By taking $|g|$ large enough, all spacings for operator $g\BV(\om)$ can be made arbitrarily
large. Eigenvectors of a  continuous finite-dimensional operator family $A(t)$ with
simple spectrum at $t=t_0$ are continuous in a neighborhood of $t_0$. To prove the second
assertion, it suffices to apply this fact to the family
$A(t) = \BW - g^{-1}t\BDelta$, $t\in[0,1]$.

The proof of the first assertion is even simpler. Using again the representation
\eqref{eq:BV.Nx}, we see that for each $\Bx$ the value of the potential energy
$\BV(\Bx;\om) + \BU(\Bx)$ is a linear combination (with integer coefficients) of
random variables with continuous probability distribution obeying \Wone. Arguing as above,
one can see that, for any $E\in\DR$, $s>0$ and $|g|$ large enough, with probability arbitrarily
close to $1$, $\dist(E, \Sigma(\BH_{\bball_{L_0}(\Bu)}))\ge s$. By Combes--Thomas estimate,
this implies exponential decay of the Green functions,
$$
|\BG_{\bball_{L_0}(\Bu)}(\Bx, \By;E)| \le e^{- m(s)\rd(\Bx,\By)}
$$
with $m(s)\to\infty$ as $s\to\infty$. Since the graph distance $\rd(\cdot\,,\cdot)$
on $\cZp$ dominates the max-distance $\rho(\cdot\,,\cdot)$, the claim follows.
\qed

\subsection{Proof of Theorem \ref{thm:PCT}}

We prove a result slightly stronger than required for the spectral analysis of
fermionic operators $\BH(\om)$: an EVC bound is established for operators
in the Hilbert space of quantum states  of distinguishable particles, not just in the
subspace of anti-symmetric functions. This simplifies geometrical arguments and notations.
So, in this subsection, we denote, for each given $n\ge 1$,  $\Bx = (x_1, \ldots, x_n)$,
$\bball_L(\Bx) = \{\By\in (\DZ^d)^n: | \By - \Bx|\le L\}$; $\BDelta$ stands for the
nearest-neighbor lattice Laplacian in $(\DZ^d)^n$, and
$\BH_{\bball_L(\Bx)} = -\BDelta_{\bball_L(\Bx)} + g\BV + \BU$. We denote by
$\rdS$ stands for the symmetrized distance:
$
\rdS(\Bx, \By) = \min_{\pi\in\fS_n} |\pi(\Bx) - \By|
$
where the elements $\pi$ of the symmetric group act by permutations of the coordinates:
$\pi(x_1, \ldots, x_n) = (x_{\pi^{-1}(1)}, \ldots, x_{\pi^{-1}(n)})$.

Given an $N$-particle configuration $\Bx$ and a proper index subset $\cJ\subset[1, N]$,
denote by $\Pi_\cJ \Bx$ the set $\{x_j, j\in\cJ\}\subset\DZ^d$.

\begin{definition}
A  cube $\bball_L(\Bx)$ is weakly separable from $\bball_L(\By)$ if
there exists a parallelepiped $Q\subset \DZ^d$ in the 1-particle configuration space,
of diameter $R \le 2NL$,  and  subsets       $\cJ_1, \cJ_2\subset [1,N]$  such that
$|\cJ_1| > |\cJ_2|$ (possibly, with $\cJ_2=\varnothing$) and
\begin{equation}\label{eq:cond.WS}
\begin{array}{l}
 \Pi_{\cJ_1} \bball_L(\Bx) \cup \Pi_{\cJ_2} \bball_L(\By) \;  \subseteq Q,\\
\Pi_{\cJ^c_2} \bball_L(\By) \cap Q = \varnothing.
\end{array}
\end{equation}
A  pair of balls $(\bball_L(\Bx), \bball_L(\By))$ is weakly separable if at least one of the balls is weakly separable from the other.
\end{definition}

\begin{lemma}[Cf. Lemma 2.3 from \cite{C10}]\label{lem:dist.are.WS}
Cubes $\bball_L(\Bx), \bball_L(\By)$ with  $\rd(\Bx, \By)> 4NL$ are weakly separable.
\end{lemma}

See the proof in  \cite{C10}.

\begin{lemma}\label{lem:lemma.PCT}

Let $V: \DZ^d\times \Omega \to \DR$ be a random field  satisfying the condition
\textbf{\rm ($\Wthree$)}. Let $\Bx,\By\in \DZ^{Nd}$ be two configurations such that the balls
$\bball_{L}(\Bx)$, $\bball_{L}(\By)$ are weakly separable.
Then for any $s>0$ the following bound holds for the  spectra
of operators $\BH_{\bball_{L'}(\Bx)}$, $\BH_{\bball_{L''}(\By)}$ with arbitrary $L', L''\le L$:
$$
\pr{ \dist(\sigma(\BH_{\bball_{L'}(\Bx)}), \sigma(\BH_{\bball_{L''}(\By)})) \le s }
     \le h_L (2s).
$$
with $h_L$ defined in \eqref{eq:h.L}.
\end{lemma}

\proof
Let $Q$ be a parallelepiped which satisfies the conditions \eqref{eq:cond.WS},
$$
\begin{array}{l}
 \Pi_{\cJ_1} \bball_L(\Bx) \cup \Pi_{\cJ_2} \bball_L(\By) \;  \subset Q,\\
\Pi_{\cJ^c_2} \bball_L(\By) \cap Q = \varnothing, \\
\end{array}
$$
for some $\cJ_1, \cJ_2 \subset [1,N]$ with  $|\cJ_1|=n_1 > n_2 =|\cJ_2|$.
In terms of the sample mean $\xi=\xi_{Q}$ and the fluctuations
$\{\eta_x, \, x\in Q \}$  defined  in subsection \ref{ssec:assumptions.V},
operators $\BH_{\bball_{L'}(\Bx)}(\omega)$,  $\BH_{\bball_{L''}(\By)}(\omega)$ read as follows:
\begin{equation}\label{eq:Ham.decomp}
\begin{array}{l}
\BH_{\bball_{L'}(\Bx)}(\omega) = n_1 \xi(\omega) \, \one + \BA(\omega), \;
\BH_{\bball_{L''}(\By)}(\omega) = n_2 \xi(\omega) \,\one + \BB(\omega)
\end{array}
\end{equation}
where operators $\BA(\omega)$ and $\BB(\omega)$ are $\fF_{V,Q}$-measurable. Let
$$
\begin{array}{l}
\{ E'_1, \ldots, E'_{M'}\}, \, \,
M' = \,\,|\bball_{L'}(\Bx)|, \\
\{ E''_1, \ldots, E''_{M''}\},
M'' = |\bball_{L''}(\By)|
\end{array}
$$
be the sets of eigenvalues of $\BH_{\bball_{L'}(\Bx)}$ and of $\BH_{\bball_{L''}(\By)})$ counted with multiplicities.  Owing to Eqn~\eqref{eq:Ham.decomp}, these eigenvalues can be represented as follows:
\def\tE{{\tilde E}}
$$
\begin{array}{l}
E'_j(\omega) = n_1\xi(\omega) + \tE'_j(\omega),\;
E''_j(\omega) = n_2\xi(\omega) + \tE''_j(\omega),
\end{array}
$$
where the random variables
$\tE'_j(\omega)$ and $\tE''_j(\omega)$ are $\fF_{V,Q}$-measurable. Therefore,
\begin{equation}\label{eq:lam.mu.xi}
E'_i(\omega) - E''_j(\omega) =  n\xi(\omega) + (\tE'_j(\omega) - \tE''_j(\omega)),
\end{equation}
with $n:=n_1-n_2 \ge 1$, owing to our assumption.
Further, we can write
$$
\begin{array}{l}
\pr{ \dist(\sigma(\BH_{\bball_{L'}(\Bx)}), \sigma(\BH_{\bball_{L''}(\By)})) \le s }
= \pr{ \exists\, i,j:\, |E'_i - E''_j| \le s } \\
\displaystyle \le \sum_{i=1}^{M'} \sum_{j=1}^{M''} \pr{ |E'_i - E''_j| \le s }
 \displaystyle = \sum_{i=1}^{M'} \sum_{j=1}^{M''}
     \esm{ \pr{ |E'_i - E''_j| \le s \,|\, \fF_{V,Q}}}.
\end{array}
$$
For all $i$ and $j$, we have
$$
\bal
\pr{ |E'_i - E''_j| \le s \,|\, \fF_{V,Q}}
&= \pr{ |(n_1 - n_2)\xi + \tE'_i - \tE''_j| \le s \,|\, \fF_{V,Q}} \\
&= \pr{ \xi \in \left[ \frac{\tE''_j - \tE'_i}{n} - \frac{s}{n},
        \frac{\tE''_j - \tE'_i  }{n} + \frac{s}{n} \right]  \,\Big|\, \fF_{V,Q}} \\
&\le \nu_L( 2 n^{-1}s \,|\, \fF_{V,Q})
\le \nu_L( 2s \,|\, \fF_{V,Q}).
\eal
$$
Therefore,
$$
\bal
\pr{ \dist(\sigma(\BH_{\bball_{L''}(\Bx)}), \sigma(\BH_{\bball_{L''}(\By)}))  \le s }
     & \le \displaystyle   M' M'' \nu_L \left( |n_1 - n_2|^{-1} s \right) \\
     &\le  \displaystyle |\bball_{L''}(\Bx)| \cdot |\bball_{L''}(\By)|\, \nu_L (2s).
     \;\; \qed
\eal
$$

Now Theorem \ref{thm:PCT} follows from Lemma \ref{lem:lemma.PCT} combined with
Lemma \ref{lem:dist.are.WS}.

\section*{Acknowledgements.}

It is a pleasure to thank Tom Spencer, Boris Shapiro, Abel Klein and Misha Goldstein for stimulating and fruitful discussions of localization techniques; the organizers of the program
\textit{"Mathematics and Physics of Anderson Localization: 50 Years After"} at the Isaac Newton Institute, Cambridge, UK (2008);  Shmuel Fishman, Boris Shapiro and the Department of Physics of Technion, Israel (2009), for their warm hospitality.


\end{document}